\documentclass[journal,twoside,web]{ieeecolor}
\usepackage{generic}
\usepackage{cite}
\usepackage{amsmath,amssymb,amsfonts}
\usepackage{algorithmic}
\usepackage{graphicx}
\usepackage{textcomp}
\def\BibTeX{{\rm B\kern-.05em{\sc i\kern-.025em b}\kern-.08em
    T\kern-.1667em\lower.7ex\hbox{E}\kern-.125emX}}

\usepackage{graphicx, color, xcolor} 
\usepackage{epsfig} 
\usepackage{amsmath, bbm, bm} 
\usepackage{amsfonts, amssymb}
\usepackage{mathrsfs}
\usepackage{anyfontsize}  
\usepackage{framed}

\usepackage{tikz}

\usepackage[framemethod=TikZ]{mdframed}
\usepackage{mathtools}
\DeclarePairedDelimiter{\ceil}{\lceil}{\rceil}

\mdfdefinestyle{theoremstyle}{backgroundcolor=black!7,roundcorner=5pt,outerlinewidth=1.1}

\mdfdefinestyle{corostyle}{roundcorner=5pt}
\mdfdefinestyle{propstyle}{roundcorner=5pt}
\mdfdefinestyle{lemmastyle}{roundcorner=5pt}
\mdfdefinestyle{defnstyle}{roundcorner=5pt}

\newcommand {\cA}{{\mathcal{A}}}
\newcommand {\cB}{{\mathcal{B}}}

\newcommand {\cT}{{\mathcal{T}}}

\newcommand {\bX} {\mathbf{X}}
\newcommand {\bY} {\mathbf{Y}}
\newcommand {\bW} {\mathbb{W}}

\newcommand {\bx} {{\bf x}}
\newcommand {\by} {{\bf y}}

\newcommand {\fttimes} {\text{\footnotesize $\times$}}
\newcommand {\onu}{\overline{\nu}}
\newcommand {\bpi} {\boldsymbol{\pi}}
\newcommand {\obpi} {\overline{\boldsymbol{\pi}}}
\newcommand {\bvarrho}{\boldsymbol{\varrho}}

\newcommand{\Eqdef}{:=}

\newcommand {\N} {\mathbb{N}}
\newcommand {\R} {{\rm I\kern-2.5pt R}}
\newcommand {\C} {{\rm I\kern-5pt C}}

\newtheorem{lemma}{Lemma}

\newtheorem{prop}{Proposition}

\newtheorem{coro}{Corollary}
\newtheorem{theorem}{Theorem}
\newtheorem{remark}{Remark}
\newtheorem{exmp}{Example}
\newtheorem{defn}{Definition}
\newtheorem{problem}{Problem}

\newcommand{\beqa}{\begin{eqnarray}}
\newcommand{\eeqa}{\end{eqnarray}}
\newcommand{\beqan}{\begin{eqnarray*}}
\newcommand{\eeqan}{\end{eqnarray*}}
\newcommand{\beq}{\begin{equation}}
\newcommand{\eeq}{\end{equation}}

\newcommand{\bfl}{\begin{flushleft}}
\newcommand{\efl}{\end{flushleft}}
\newcommand{\bgnv}{\begin{verbatim}}
\newcommand{\ednv}{\end{verbatim}}

\newcommand{\myb}{\hspace{-0.1in}}
\newcommand{\myeq}{& \hspace{-0.1in} = & \hspace{-0.1in}}

\newcommand{\lb}{\nonumber \\}
\newcommand{\myarr}{\begin{array}{lll}}

\newcommand{\mygeq}{& \myb \geq & \myb}

\newcommand{\myleq}{& \myb \leq & \myb}

\newcommand{\bitem}{\begin{itemize}}
\newcommand{\eitem}{\end{itemize}}
\newcommand{\benum}{\begin{enumerate}}
\newcommand{\eenum}{\end{enumerate}}

\newcommand{\norm}[1]{\left| \left| #1 \right| \right|}

\newcommand{\myhb}{\hspace{-0.2in}}
\newcommand{\myhf}{\hspace{0.3in}}

\newcommand{\myskip}{\\ \vspace{-0.1in}}

\newcommand{\oby}{\overline{{\bf y}}}
\newcommand{\os}{\overline{s}}
\newcommand{\ow}{\overline{w}}
\newcommand{\opi}{\overline{\pi}}

\newcommand{\obY}{\overline{\bf Y}}
\DeclarePairedDelimiter\floor{\lfloor}{\rfloor}

\DeclareMathOperator{\EX}{\mathbb{E}}

\DeclareMathOperator*{\plim}{plim}

\def\QED{~\rule[-1pt]{5pt}{5pt}\par\medskip}

\newcommand{\remind}[1]{{\bf {\color{red}[\marginpar[\hbox{{\large%
$\circ$}\raisebox{0ex}{\large $\longrightarrow$}}]%
{\hbox{\raisebox{0ex}{\large $\longleftarrow$}{\large
$\circ$}}}#1]} } }

\let\labelindent\relax

\usepackage{enumitem}
\usepackage{chngcntr}

\title{\LARGE \bf
Queueing Subject To Action-Dependent Server Performance: \\ Utilization Rate Reduction
}

\author{Michael Lin \and Nuno C. Martins \and Richard J. La \thanks{The authors are with the 
Department of Electrical \& Computer 
Engineering and 
the Institute for Systems Research, the
University of Maryland, College Park, MD 20742. 
Email: \{mlin1025, hyongla, nmartins\}@umd.edu.}
\thanks{This work is supported in part by 
AFOSR Grant FA95501510367 and NSF Grant ECCS 1446785.}
}

\begin{document}

\maketitle
\thispagestyle{empty}
\pagestyle{plain}

\begin{abstract}

We consider a discrete-time system comprising a first-come-first-served queue, a non-preemptive server, 
and a stationary non-work-conserving scheduler. New tasks enter the queue according to a Bernoulli process with a pre-specified arrival rate. At each instant, the server is either busy working 
on a task or is available. When the server is available, the scheduler 
either assigns a new task to the server or allows it to remain available (to \emph{rest}). In addition to the aforementioned {\em availability} 
state, we assume that the server has an integer-valued {\em 
activity state}.
The  activity state is non-decreasing during work 
periods, and is non-increasing otherwise.
In a typical application of our framework, the server performance (understood as task completion probability) worsens as the activity state increases. In this article, we build on and transcend recent stabilizability results obtained for the same framework. Specifically, we establish methods to design scheduling policies that not only stabilize the queue but also reduce the {\em utilization rate} \textemdash understood as the infinite-horizon time-averaged portion of time the server is working. This article has a main theorem leading to two key results: (i)~We put forth a tractable method to determine, using a finite-dimensional linear program (LP), the infimum of all utilization rates that can be achieved by scheduling policies that are stabilizing, for a given arrival rate. (ii)~We propose a design method, also based on finite-dimensional LPs, to obtain stabilizing scheduling policies that can attain a utilization rate arbitrarily close to the aforementioned infimum. We also establish structural and distributional convergence properties, which are used throughout the article, and are  significant in their own right.



\end{abstract}

\section{Introduction}
In this article, we adopt the discrete-time framework proposed in~\cite{Lin2019Scheduling-task}, in which a scheduler governs when tasks waiting in a first-come-first-served queue are assigned to a server. The server is non-preemptive, and has an internal state comprising two components: (i) the \emph{availability state} and (ii) \emph{activity state}. The former indicates whether the server is busy or available, and the latter takes values in a finite set $\{1,\ldots,n_s\}$ that accounts for the intensity of the effort put in by the server. The activity state depends on current and previous scheduling decisions, and it is useful for modelling performance-influencing factors, such as the state of charge of the batteries of an energy harvesting module that powers one or more components of the server. As a rule, the activity state may increase while the server is busy and, otherwise, decrease gradually while the server is available ({\em or resting}). 

In our framework, which follows~\cite{Lin2019Scheduling-task}, an {\it instantaneous service rate function} ascribes to each possible activity state a probability that the server can complete a task in one time-step. According to our assumption of non-preemption, once the server becomes busy working on a task, it becomes available again only when the task is completed. When the server is available, the scheduler decides,  based on the activity state and the size of the queue, whether to assign a new task to the server. Although our results remain valid for any instantaneous service rate function, in many applications it is decreasing, which causes the server performance (understood as task completion probability) to worsen as the activity state increases. The vital trade-off the scheduler faces, in this case, is whether to assign a new task when the server is available or allow it to remain available~(rest) to possibly ameliorate the activity state as a way to improve future performance.


\subsection{Problem Statements and Comparison to~\cite{Lin2019Scheduling-task}}

Besides introducing and justifying in detail the formulation adopted here, in~\cite{Lin2019Scheduling-task} the authors characterize the supremum of all arrival rates for which there is a scheduler that can stabilize the queue. The analysis in~\cite{Lin2019Scheduling-task} also shows that such a supremum can be computed by a finite search, and identifies simple stabilizing scheduler structures, such as those with a threshold-type configuration.

In this article, we build on the analysis in~\cite{Lin2019Scheduling-task} to design schedulers that not only guarantee stability but also lessen the \underline{utilization rate}, which we will define precisely later on and can be interpreted as the proportion of time in which the server is working. Specifically, throughout this article, we will investigate and provide solutions to the following two problems.

\begin{problem}\label{Problem:ComputeBound} Given a server and a stabilizable arrival rate\footnote{A given arrival rate is deemed {\it stabilizable} when there is a scheduling policy for which the queue is stable in the sense specified in~\cite{Lin2019Scheduling-task} and that will be precisely defined also in this article later on.}, determine a tractable method to compute the infimum of all utilization rates that can be achieved by a stabilizing scheduling policy. Such a fundamental limit is important to determine how effective any given stabilizing policy is in terms of the utilization rate.
\end{problem}

\begin{problem}\label{Problem:ComputePolicy} Given a server and a stabilizable arrival rate, determine a tractable method to design stabilizing scheduling policies whose utilization rate is arbitrarily close to the fundamental limit. 
\end{problem}


\subsection{Overview of Main Results and Technical Approach}

In \S\ref{sec:MainResults}, Theorem~\ref{thm:PolicyDesign} states our main result, from which we obtain Corollaries~\ref{cor:Equality} and~\ref{cor:PolicyDesign} that constitute our solutions to Problems~1 and~2, respectively. The following are key consequences of these corollaries. (i)~According to Corollary~\ref{cor:Equality}, the infimum utilization rate (alluded to in Problem~1) can be computed by solving a finite-dimensional linear program (LP). (ii)~If the arrival rate is stabilizable by the server, then Corollary~\ref{cor:PolicyDesign} guarantees that, for each positive gap $\delta$, there is a stabilizing scheduling policy whose utilization rate exceeds the infimum (characterized by Corollary~\ref{cor:Equality}) by at most $\delta$. Notably, such a scheduling policy can be obtained from a solution of a suitably-specified finite-dimensional LP. 

Our technical approach builds on the concepts and techniques introduced in~\cite{Lin2019Scheduling-task}. In particular, we use an appropriately-constructed auxiliary finite-state controlled Markov chain (denoted in~\cite{Lin2019Scheduling-task} as \emph{reduced process}) to obtain the above-mentioned LP-based solution methods. 

This article is mathematically more intricate than~\cite{Lin2019Scheduling-task}, which is unsurprising considering that it tackles not only stabilization but also regulation of the utilization rate. Among the new concepts and techniques put forth to prove Theorem~\ref{thm:PolicyDesign}, the distributional convergence results of \S\ref{sec:DistributionalConvergence}, and the potential-like method used to establish them, are of singular importance \textemdash they are also original and relevant in their own right.

\subsection{Related Literature}

As mentioned earlier, to the best of our knowledge, 
our work is the first to study the problem of 
lessening the utilization rate of a server whose
performance is time-varying and dependent on 
an internal state that reflects its activity
history. For this reason, there are no 
other results to which we can directly compare our findings. 

An earlier study that examined a system that closely
resembles ours is that of Savla and Frazzoli~\cite{Savla2012A-Dynamical-Que}. They studied
the problem of designing a maximally stabilizing
task release control policy, using a differential
system model. Under an assumption that the service time function is convex, they 
derived bounds on the maximum 
throughput achievable by any admissible
policy for a fixed task workload
distribution. 
In addition, they showed the existence of 
a maximally stabilizing threshold policy
when the tasks have the same workload. 
Finally, they demonstrated that
the maximum achievable throughput 
increases when the task workload
is not deterministic. 
However, they did not consider the problem 
of minimizing utilization rate in 
their study. 

In addition to the aforementioned
study, there are a few research
fields that share a key aspect of
our problem, which is to design a scheduling
policy to optimize the performance with
respect to one objective, subject to one or
more constraints. For instance, wireless
energy transfer has emerged as a potential
solution to powering small devices that
have low-capacity batteries or cannot be
easily recharged, e.g., Internet-of-Things
(IoTs) devices~\cite{Bi16, Niyato17}. 
Since the devices need to collect sufficient
energy before they can transmit and 
the transmission rate is a function of 
transmit power, a transmitter has to 
decide (i) when to harvest energy and 
(ii) when to transmit and at what 
rate. For example, the
studies reported in \cite{Che15, Ju14, Yang15} 
examined the problem of maximizing throughput
in wireless networks in 
which communication devices are powered
by hybrid access points via
wireless energy transfer. In a related
study, Shan et al.~\cite{Shan16} studied
the problem of minimizing the total 
transmission delay or completion time
of a given set of packets.

Integrated production scheduling and 
(preventive) maintenance planning in 
manufacturing, where machines can fail 
with time-varying rates, shares similar
issues as scheduling devices powered
by wireless energy transfer
\cite{Cassady05, Najid11, Yao05}. 
In more traditional approaches, the problems 
of production scheduling and maintenance 
scheduling are considered separately, and
equipment failures are treated as random 
events that need to be coped with. 
When the machine failure probability, or
rate, is time-varying and depends on the 
length of time (age) elapsed
since the last (preventive) maintenance, 
the overall production efficiency can 
be improved by jointly considering both
problems. For instance, the authors of
\cite{Yao05} formulated the problem 
using an MDP model with the state 
consisting of the system's age (since the
last preventive maintenance) and the
inventory level, and investigated the 
structural properties of optimal policies. 

Another area that shares a similar objective
is the maximum hand-offs control or
sparse control~\cite{Nagahara16, Chatterjee16,
Ikeda16, Ikeda16b, Ikeda19}. 
The goal of the maximum hands-off
control is to design a control signal that
maximizes the time at which the control 
signal is equal to zero and inactive. 
For instance, the authors of
\cite{Nagahara16} showed that, under the
normality condition, the optimal solution 
sets of a maximum hands-off control 
problem and an associated $L^1$-optimal
control problem coincide. Moreover, 
they proposed a self-triggered feedback
control algorithm for infinite-horizon
problems, which leads to a control signal
with a provable sparsity rate, while
achieving practical stability of the 
system. In another study~\cite{Chatterjee16}, 
Chatterjee et al. provided both 
necessary conditions and sufficient
conditions for maximum hands-off
control problem. Ikeda and Nagahara
\cite{Ikeda16} considered a linear
time-invariant system and showed 
that, if the system is 
controllable and the dynamics matrix
is nonsingular, the optimal value
of the optimal control problem
for the maximum hands-off control
is continuous and convex in the 
initial condition.

Finally, another research problem, which garnered
much attention in wireless sensor networks 
and is somewhat related to the maximum 
hands-off control, 
is duty-cycle scheduling of sensors. A common
objective for the problem 
is to minimize the total energy 
consumption subject to performance 
constraints on delivery reliability
and delays \cite{Ergen08}. The authors
of \cite{Liu06} proposed using a 
reinforcement learning-based control mechanism
for inferring the states of neighboring sensors
in order to minimize the active periods. 
In another study, Vigorito et al. studied
the problem of achieving energy neutral
operation (i.e., keep the battery charge
at a sufficient level) 
while maximizing the awake times
\cite{Vigorito07}. In order to design a
good control policy, they formulated the
problem as an optimal tracking problem, 
more precisely a linear quadratic tracking
problem, with the aim of keeping the
battery level around some target value.

\subsection{Paper Structure}

This article has five sections. After the introduction, in~\S\ref{sec:TechnicalFramework}, we describe the technical framework, including the controlled Markov chain that models the server. In~\S\ref{sec:TechnicalFramework}, we also introduce a relevant auxiliary {\em reduced} process, define key quantities and maps that quantify the utilization rate, characterize key policy sets, specify the notion of stability used throughout the article, and establish certain preliminary results. Our main theorem and key results are stated in~\S\ref{sec:MainResults}, while~\S\ref{sec:ContinuityProperties} and~\S\ref{sec:DistributionalConvergence} present continuity and distributional convergence properties, respectively, that are required in the proof of our main theorem. We defer the most intricate proofs, some of which also require additional auxiliary results, to the appendices at the end of the article. The main body of the article ends with brief conclusions in~\S\ref{sec:conclusions}.

\section{Technical Framework and Key Definitions}
\label{sec:TechnicalFramework}


 This section starts with a synopsis of the discrete-time framework put forth thoroughly in~\cite{Lin2019Scheduling-task}. Henceforth, we replicate from~\cite{Lin2019Scheduling-task} what is strictly necessary to make this article self-contained. In this section, we also introduce the concepts, sets, operators and notation that are required to formalize and later on solve Problems~\ref{Problem:ComputeBound} and~\ref{Problem:ComputePolicy}.


\begin{remark} According to the approach in~\cite{Lin2019Scheduling-task}, each discrete-time~$k$ represents a continuous-time interval, or epoch, whose duration can be made arbitrarily small. Considering that this representation is described in detail in~\cite{Lin2019Scheduling-task}, here we proceed directly to the description of the resulting discrete-time framework and we refer to each epoch $k$ simply as time~(instant) $k$, with $k$ in $\mathbb{N}:=\{0,1,2,\ldots\}$.
\end{remark}

\subsection{Stochastic Discrete-Time Framework}

As in~\cite{Lin2019Scheduling-task}, we consider that the server is represented by the MDP $\bY:=\{\bY_k \in \mathbb{Y} : k \in \mathbb{N} \}$. The state of the server at time $k$ is $\bY_k :=(S_k,W_k)$, whose components are the activity state $S_k$ and the availability state $W_k$ taking values in $\mathbb{S}:=\{1,\cdots,n_s\}$ and $\mathbb{W}:=\{\mathcal{A},\mathcal{B}\}$, respectively. Here, $W_k=\mathcal{A}$ indicates that the server is available at time $k$, while $W_k=\mathcal{B}$ signals that the server is busy. Consequently, the state-space of the server is represented as
\begin{equation}
\mathbb{Y} := \mathbb{S} \times \mathbb{W}. 
\end{equation}

The MDP $\bX := \{ \bX_k \in \mathbb{X} : k \in \mathbb{N} \}$ represents the overall system comprising the server $\bY$ and the queue length. More specifically, the state of the system is $\bX_k : = (\bY_k,Q_k)$, where $Q_k$ is the length of the queue at time $k$, and the state-space of $\bX$ is:
\begin{equation}
\label{eq:XDef}
\mathbb{X} : = \mathbb{S} \times \Big ( \big ( \mathbb{W} \times \mathbb{N} \big) \diagdown (\mathcal{B},0) \Big )
\end{equation}
Notice that $\mathbb{X}$ excludes the impossible case in which the server would be busy working with an empty queue.

The action of the scheduler at time $k$ is represented by $A_k$, which takes values in the set $\mathbb{A} : = \{ \mathcal{R},\mathcal{W} \}$. The scheduler directs the server to work at time $k$ when $A_k = \mathcal{W}$ and instructs the server to rest when $A_k = \mathcal{R}$. Since the server is non-preemptive, once it is busy working on a task it is not allowed to rest until the task is completed and it becomes available again.
This constraint and the fact that 
no new tasks can be assigned when the queue is empty, lead to the 
following set of admissible actions for each possible state 
${\bx=(s,w,q)}$ in~$\mathbb{X}$: 
\begin{equation}
\label{ActionConstraints}
\mathbb{A}_\bx
:= \begin{cases}
	\{\mathcal{R}\} & \text{if  $q = 0, \  $ {\small (impose 
		`rest' when queue is empty)}} \\
    \{\mathcal{W}\} & \text{if  $q > 0$ and  $w = \mathcal{B}$, \ {\small 
    	(non-preemptive server)}} \\
    \mathbb{A} & \text{otherwise}. \\
	\end{cases}
\end{equation} 

We assume that tasks arrive according 
to a Bernoulli process $\{B_k; \ k \in \N\}$. 
The {\em arrival rate} is denoted by $\lambda:
= {P(B_k=1)}$.

\subsubsection{\underline{Activity-Dependent Server Performance}}
  
In our formulation, the efficiency or performance of the 
server is modeled with the help of 
an {\em instantaneous service rate} function $\mu: \mathbb{S} \to 
(0, 1)$. More specifically, if the server works on a task
at time $k$, \underline{the probability} that it 
completes the task before time $k+1$ is $\mu(S_k)$. 
This holds irrespective of whether the task is newly 
assigned or inherited as ongoing work.
Thus, $\mu$ quantifies the effect 
of the activity state on the performance of the 
server. {\bf The results presented throughout this 
article are valid for {\em any} choice of $\mu$ with 
codomain~$(0,1)$}.

\subsubsection{\underline{Dynamics of the activity state}}

We assume that (i)~$S_{k+1}$ is equal to either $S_k$ 
or $S_k+1$ when $A_k$ is 
$\mathcal{W}$ and (ii)~$S_{k+1}$ is either 
$S_k$ or $S_k - 1$ if $A_k$ 
is $\mathcal{R}$.  The state-transition probabilities for $S_k$ are specified below for every $s$ and $s'$ in $\mathbb{S}$:
\begin{subequations}
\label{Def-SDynamics}
\begin{align}
 P_{S_{k+1} | S_k, A_k}(s' \ | \ s, \mathcal{W}) & =
\begin{cases}
	\rho_{s,s+1} & \mbox{if }  
    	s' = s + 1 \\
    1 - \rho_{s, s+1} & \mbox{if } 
    	s' = s \\
    0 & \mbox{otherwise}
	\end{cases} \\
 P_{S_{k+1} | S_k, A_k}(s' \ | \ s, \mathcal{R}) &= 
\begin{cases}
	\rho_{s,s-1} & \mbox{if }
    	s' = s-1 \\
    1 - \rho_{s, s-1} & \mbox{if } 
    	s' = s \\
    0 & \mbox{otherwise}
	\end{cases}
\end{align}
\end{subequations} where the parameters $\rho_{s,s'}$ quantify the likelihood that the activity 
state will transition to a greater or lesser value, depending on 
whether the action is $\mathcal{W}$ or $\mathcal{R}$, respectively.
Here, we assume that $\{\rho_{s,s+1} : 1\leq s < n_s \}$ and $\{\rho_{s,s-1}: 1< s \leq n_s \}$ take values in $(0,1)$. We also adopt the convention that $\rho_{1,0}=\rho_{n_s,n_s+1}=0$.

\subsubsection{\underline{Transition probabilities for $\bX_k$}}

We consider that $S_{k+1}$ is independent of {${(W_{k+1}, Q_{k+1})}$} 
when conditioned on {{$(\bX_k, A_k)$}}. Under this assumption, the 
transition probabilities for $\bX_k$ can be written as follows: 
\beqa
\label{Eqdef:MDP}
&& \myhb
P_{\bX_{k+1}|\bX_k,A_k}( \bx'  \ | \ \bx, a) \lb
\myeq  P_{S_{k+1} | \bX_k, A_k}(s' \ | \ \bx, a) \lb 
&& \times P_{W_{k+1},Q_{k+1} | \bX_k, A_k}(w', q' \ | \ 
	\bx, a) \lb 
\myeq P_{S_{k+1} | S_k, A_k}(s' \ | \ s, a) 
	\label{eq:Xt+1} \\
&& \times P_{W_{k+1},Q_{k+1} | \bX_k, A_k}(w', q' \ | \ 
	\bx, a)
	\nonumber
\eeqa for every $\bx$, $\bx'$ in $\mathbb{X}$ and  $a$ in $\mathbb{A}_{\bx}$.

We assume that, at each time $k$, the events that
(i) there is a new task arrival and (ii) a 
task being serviced is completed are independent when conditioned on $\bX_k$ and 
$\{A_k=\mathcal{W}\}$. Hence, the transition probability 
$P_{W_{k+1},Q_{k+1} | \bX_k, A_k}$ in (\ref{eq:Xt+1}) is 
given by the following:
\begin{subequations}
\label{WTransitionDef}
\beqa
&& \myhb P_{W_{k+1},Q_{k+1} | \bX_k, A_k}(w', q' 
	\ | \ \bx, \mathcal{W}) \label{eq:WQW} \\
\myeq \left\{ \begin{array}{ll}
	\mu(s) \ \lambda & \mbox{if } w' = \cA \mbox{ and }
    	q' = q, \\
	\mu(s) \ (1-\lambda) & \mbox{if } w' = \cA \mbox{ and } 
    	q' = q - 1, \\
	(1-\mu(s)) \ \lambda & \mbox{if } w' = \cB \mbox{ and }
    	q' = q + 1, \\
	(1-\mu(s)) (1-\lambda) & \mbox{if } w' = \cB 
    	\mbox{ and } q' = q, \\
	0 & \mbox{otherwise, } \\
	\end{array} \right. \lb 
&& \myhb P_{W_{k+1},Q_{k+1} | \bX_k, A_k}(w', q' \ | \ 
	\bx, \mathcal{R}) \label{eq:WQR} \\
\myeq \left\{ \begin{array}{ll}
	\lambda & \mbox{if } w' = \cA \mbox{ and } q' = q+1, \\
	1-\lambda & \mbox{if } w' = \cA \mbox{ and } q' = q, \\
	0 & \mbox{otherwise.} \\
	\end{array} \right.
    \nonumber
\eeqa
\end{subequations}

\begin{defn}{\bf (MDP $\bX$)} The MDP with input $A_k$ and state $\bX_k$, which at this point is completely defined, is denoted by~$\bX$. 
\end{defn}

Table~\ref{table:Notation} summarizes the notation for MDP~$\bX$.
\begin{table}[h]
\begin{center}
\begin{tabular}{c|l} \hline
$\mathbb{S}$ & set of activity states $\{1, 	
	\ldots, n_s\}$ \\
$\bW \Eqdef \{\cA, \cB\}$ & server availability ($\cA =$ 
	available, $\cB =$ busy) \\
$W_k$ & server availability at time $k$ (takes values in 
	$\mathbb{W}$) \\
$\mathbb{Y}$ & server state components $\mathbb{S}\times
	\mathbb{W}$\\
$\bY_k \Eqdef (S_k, W_k)$ & server state at time $k$ 
	(takes values in $\mathbb{Y}$) \\
$\mathbb{N}$ & natural number system $\{0,1,2,\ldots \}$. \\
$Q_k$ & queue size at time $k$ (takes values in 
	$\mathbb{N}$) \\
$\mathbb{X}$ & state space formed by $\mathbb{S}\times 
	\Big( (\mathbb{W}\times\mathbb{N})\diagdown 
	(\mathcal{B},0) \Big)$ \\
$\mathbf{X}_k \Eqdef (\bY_k, Q_k)$ & system state at 
	time $k$ (takes values in $\mathbb{X}$)\\
$\mathbf{X}$ & MDP whose state is $\mathbf{X}_k$ at time 
	$k \in \N$ \\
$\mathbb{A} \Eqdef \{\mathcal{R},\mathcal{W}\}$ & possible 
	actions ($\mathcal{R}$ = rest, $\mathcal{W}$ = work) \\
$\mathbb{A_\mathbf{x}}$ & set of actions admissible at a 
	given state $\mathbf{x}$ in $\mathbb{X}$\\
$A_k$ & action chosen at time $k$. \\
PMF & probability mass function \\
\hline
\end{tabular}
\vspace{.1 in}
\caption{A summary of notation describing MDP $\bX$.}
\label{table:Notation}
\end{center}
\end{table}

\subsubsection{Stationary Policies, Stability and Stabilizability}

We start by defining the class of policies that we consider 
throughout the paper. 

\begin{defn} A stationary randomized policy is specified by
a mapping $\theta: \mathbb{X} \to [0, 1]$ that determines the 
probability that the server is assigned to work on a task or rest, 
as a function of the system state, according to
\begin{align*}
P_{A_k|\mathbf{X}_k,\ldots,\mathbf{X}_0}(\mathcal{W} \ | \ x_k,\ldots,x_0) 
	&= \theta(x_k) \ \mbox{ and } \\
P_{A_k|\mathbf{X}_k,\ldots,\mathbf{X}_0}(\mathcal{R} \ | \ x_k,\ldots,x_0) 
	&= 1-\theta(x_k).
\end{align*}
\\ \vspace{-0.35in}
\end{defn}

\begin{defn}	\label{def:PhiR}
The set of stationary randomized policies satisfying 
(\ref{ActionConstraints}) is denoted by $\Theta_R$.  
\myskip
\end{defn}

\paragraph*{Convention} Although the statistical properties of 
$\bX$ subject to a 
given policy depend on the parameters specifying $\bX$, including $\lambda$, we simplify our notation by not representing this dependence, unless noted otherwise. With the exception of $\lambda$, which we think of as a variable, we assume that all the other parameters for $\bX$ are given and fixed throughout the paper.

From (\ref{eq:Xt+1}) - (\ref{eq:WQR}), we conclude that $\mathbf{X}$ 
subject to a policy $\theta$ in  $\Theta_R$ evolves according to a
time-homogeneous Markov chain (MC), which we denote by
$\bX^\theta := \{\bX^\theta_k : k \in \N \}$.
Also, provided that it is clear from the context, 
we refer to $\bX^\theta$ as {\em the system}. 


The following is the notion of system stability we adopt 
throughout this article.
\myskip 

\begin{defn}[System stability, stabilizability and $\Theta_S(\lambda)$] \label{def:Stability} 
For a given policy $\theta$ in $\Theta_R$, 
the system $\bX^\theta$ is stable if it satisfies the following 
properties: 

\noindent (i)~The number of transient states is finite and, hence, 
there is at least one recurrent communicating class. \\
(ii)~All recurrent communicating classes are positive recurrent.

An arrival rate $\lambda$ is said to be stabilizable when there is a policy $\theta$ in $\Theta_R$ for which $\bX^\theta$ is stable. We also define $\Theta_S(\lambda)$ to 
be the set of randomized policies in $\Theta_R$ that
stabilize the system for a stabilizable arrival rate $\lambda$.
\end{defn}

Before we proceed, let us point out a useful fact
under any stabilizing policy $\theta$ in 
$\Theta_S(\lambda)$. 
 
\begin{lemma} \cite[Lemma~1]{Lin2019Scheduling-task}	\label{lemma:UniquePMF}
A stable system $\bX^\theta$ has a unique
positive recurrent communicating class, which
is aperiodic. Therefore, there is a unique 
stationary probability mass function (PMF) 
for $\bX^\theta$. 
\end{lemma}

\begin{defn} \label{def:pi}
Given an arrival 
rate $\lambda > 0$ and a stabilizing policy 
$\theta$ in $\Theta_S(\lambda)$, we denote the unique 
stationary PMF and positive recurrent communicating
class of $\bX^\theta$ by ${\bpi^{\theta} = 
(\pi^{\theta}(\bx) : \bx \in \mathbb{X})}$ 
and $\mathbb{C}_\theta$, respectively.
\end{defn}

\subsection{Utilization Rate: Definition and Infimum}


Subsequently, we proceed to define the concepts and maps required to formalize the analysis and computation of the utilization rate, and its infimum alluded to in the statements of Problems~\ref{Problem:ComputeBound} and~\ref{Problem:ComputePolicy}.

\begin{defn}{\bf (Utilization rate function)} The function that determines the utilization rate in terms of a given stabilizable arrival rate $\lambda$ and a stabilizing policy $\theta$, is defined as:
\begin{equation*}
\mathscr{U}(\lambda,\theta) 
:= \sum_{ \bx \in \mathbb{X}} \pi^{\theta}(\bx) 	
	\theta(\bx), 
	\quad \lambda \in (0,\lambda^*), 
	\ \theta \in \Theta_S(\lambda)
\end{equation*}
\end{defn}
The utilization rate quantifies the probability that the server is working, in the stationary limit. Notably, $\mathscr{U}(\lambda,\theta)$, computed for $\bX$ with arrival late $\lambda$ and stabilized by $\theta$, coincides with the probability limit of the utilization rate, as defined for instance in~\cite{Jog2019Channels-learni} (with $\mathbb{U}=\{0,1\}$), when the averaging horizon tends to infinity. Using our notation, the aforesaid probability limit can be stated as follows:
\begin{multline*}
\myb \plim_{N \rightarrow \infty} 
	\frac{\sum_{k=0}^N \mathcal{I}_{A_k = \mathcal{W}}}
		{N+1} 
= \mathscr{U}(\lambda,\theta),
\quad \lambda \in (0,\lambda^*) , \ \theta \in \Theta_S(\lambda) 
\end{multline*} 
where $\mathcal{I}_{A_k = \mathcal{W}}$ is $1$ 
when $A_k = \mathcal{W}$ and $0$ otherwise. Hence, the utilization rate can also be viewed as the proportion of time in which the server is working, in the infinite time-horizon limit.

\begin{defn} The infimum utilization rate for a given stabilizable arrival rate $\lambda$ is defined as
\begin{equation*}
\mathscr{U}^*(\lambda) 
: = \inf_{\theta \in \Theta_S(\lambda)} 
	\mathscr{U}(\lambda,\theta), 
	\quad \lambda \in (0,\lambda^*).
\end{equation*}
\end{defn}

\subsection{Auxiliary MDP $\overline{\mathbf{Y}}$}

We proceed with describing an underlying controlled Markov chain
whose state takes values in $\mathbb{Y}$ 
and approximates the server state of $\mathbf{X}$ 
under a subclass of policies in $\Theta_R$, subject to an 
assumption that the queue always has a task to service
whenever the server becomes available. We  
endow it with a reward function, which 
is the utilization of the server, and denote 
this auxiliary MDP by $\overline{\mathbf{Y}}$ and 
its state at time $k$ by $\overline{\bY}_k
:=(\overline{S}_k,\overline{W}_k)$ in order to 
emphasize that it takes values in~$\mathbb{Y}$. 
The action chosen at time $k$ is denoted by 
$\overline{A}_k$. We use the overline to denote 
the auxiliary MDP and any other variables 
associated with it, in order to distinguish them 
from those of the server state in $\bX$. 

Under certain conditions, which we will identify later on, we can determine important properties of  $\bX$ by analysing~$\overline{\bY}$. Notably, we will use the fact that $\mathbb{Y}$ is finite to compute $\mathscr{U}^*$ via a finite-dimensional LP, and also to simplify the proofs of our main results. 

As the queue size is no longer a component of the state of~$\overline{\bY}$,  we eliminate the dependence of the admissible
action sets on $q$, which was explicitly specified in 
(\ref{ActionConstraints}) for MDP $\bX$, while 
still ensuring that the server is non-preemptive. More 
specifically, the set of admissible actions at each 
element ${\overline{\by}:=(\overline{s},\overline{w})}$ of 
$\mathbb{Y}$ is given by 
\begin{equation}
\label{ActionConstraintsauxiliary}
\overline{\mathbb{A}}_{\overline{w}}
\Eqdef \begin{cases}
    \{\mathcal{W}\} & \text{if $\overline{w} = \cB$, \quad 
    	(non-preemptive server)} \\
    \mathbb{A} & \text{if $\overline{w} = \cA$}. \\
	\end{cases}
\end{equation} 
Consequently, for any given realization of the current state $
\overline{\by}_k = (\overline{s}_k,\overline{w}_k)$, $\overline{A}_k$ 
is required to take values in $\overline{\mathbb{A}}_{\overline{w}_k}$.

We define the transition probabilities that specify $\overline{\bY}$, 
as follows:
\beqa
P_{\overline{\bY}_{k+1}|\overline{\bY}_k,\overline{A}_k}( \overline{\by}' 
 	 \ | \ \overline{\by}, \overline{a}) 
& \myb \Eqdef & \myb 
P_{\overline{S}_{k+1} | \overline{S}_k, \overline{A}_k}(\overline{s}' 
	\ | \ \overline{s}, \overline{a}) 
	\label{Eqdef:MDPRED} \\ 
&& \fttimes P_{\overline{W}_{k+1} | 
		\overline{\bY}_k, \overline{A}_k}(\overline{w}' \ |\ \overline{\by}, 
			\overline{a}), 
	\nonumber
\eeqa
where $\overline{\by}$ and $\overline{\by}'$ are in $\mathbb{Y}$, and 
$\overline{a}$ is in $\overline{\mathbb{A}}_{\overline{w}}$. The right-hand terms of (\ref{Eqdef:MDPRED}) 
are defined, in connection with $\bX$, as follows:
\begin{equation}
\label{SDef}
P_{\overline{S}_{k+1} | \overline{S}_k, \overline{A}_k}(\overline{s}' 
	\ | \ \overline{s}, \overline{a}) 
\Eqdef P_{S_{k+1} | S_k, A_k}(\overline{s}' \ | \ \overline{s}, \overline{a})  \\
\vspace{-0.1in}
\end{equation}
\begin{subequations}
\label{eq:REDWTran}
\begin{align}
P_{\overline{W}_{k+1} | \overline{\bY}_k, \overline{A}_k}(\overline{w}' 
	\ | \ \overline{\by}, \mathcal{W}) 
\Eqdef & \begin{cases} 
	\mu(\overline{s}) & \text{if $\overline{w}'=\mathcal{A}$} \\ 
	1-\mu(\overline{s}) & \text{if $\overline{w}'=\mathcal{B}$} 
	\end{cases} \\
P_{\overline{W}_{k+1} | \overline{\bY}_k, \overline{A}_k}(\overline{w}' 
	\ | \ \overline{\by}, \mathcal{R})
\Eqdef & \begin{cases} 
	1 & \text{if $\overline{w}'=\mathcal{A}$} \\ 
	0 & \text{if $\overline{w}'=\mathcal{B}$} 
	\end{cases}
\end{align}
\end{subequations}


From~(\ref{eq:REDWTran}) and (\ref{WTransitionDef}), 
we can deduce 
the following equality valid for all $q \geq 1$, 
\beqa
&& \myhb P_{\overline{W}_{k+1} | \overline{\bY}_k, \overline{A}_k} 
	\big( \overline{w}' 
	\ | \ \overline{\by}, \mathcal{W} \big) \lb
\myeq \sum_{q' = 0}^\infty P_{W_{k+1},Q_{k+1} | \bX_k, A_k} 
	\big( (\overline{w}',q') \ | \ (\ \overline{\by},q), \mathcal{W} 
		\big),
	\label{equivalenceXandY-W} 
\eeqa
which holds for any $\overline{w}'$ in $\mathbb{W}$ and $\overline{\by}$ 
in $\mathbb{Y}$. Notice that the right-hand side (RHS) of 
(\ref{equivalenceXandY-W}) does not change when we vary $q$ across the 
positive integers. From this, in conjunction with 
(\ref{eq:Xt+1}), (\ref{Eqdef:MDPRED})
and (\ref{SDef}), we also have, for all $q \geq 1$, 
\beqa
&& \myhb P_{\overline{\bY}_{k+1} | \overline{\bY}_k, \overline{A}_k} 
	\big(\overline{\by}' \ | \ \overline{\by}, \mathcal{W}) \lb 
\myeq \sum_{q'=0}^{\infty} P_{\bX_{k+1} | \bX_k, A_k} 
	\big( (\overline{\by}',q') \ | \ ( \overline{\by},q), \mathcal{W} 
		\big).
	\label{equivalenceXandY} 
\eeqa
The equality in (\ref{equivalenceXandY}) indicates that 
$P_{\overline{\bY}_{k+1} | \overline{\bY}_k, \overline{A}_k}$ 
also characterizes the transition probabilities of 
the server state $\bY_k = (S_k, W_k)$ in $\bX$ 
when the current queue size is 
positive. This is consistent with our earlier viewpoint that 
$\overline{\bY}$ {\it behaves} as the server state in $\bX$ when the queue is nonempty. 

\subsection{Stationary Policies and Stationary PMFs of $\overline{\mathbf{Y}}$}

Analogously to the MDP $\bX$, we only consider 
stationary randomized policies for $\overline{\bY}$, 
which are defined below.  
\myskip

\begin{defn}[$\Phi_R$] 
\label{def:StationaryPoliciesForbY}
We restrict our attention to stationary randomized policies acting on 
$\overline{\bY}$, which are specified by 
a mapping ${\phi: \mathbb{Y} \to [0, 1]}$, as follows:
\beqan
P_{\overline{A}_k | \overline{\bY}_k,\ldots,\overline{\bY}_0 } 
	(\mathcal{W} \ | \ \overline{\by}_k,\ldots,\overline{\by}_0)
\myeq \phi(\overline{\by}_k) \\
P_{\overline{A}_k | \overline{\bY}_k,\ldots,\overline{\bY}_0 } 
	(\mathcal{R} \ | \ \overline{\by}_k,\ldots,\overline{\by}_0)
\myeq 1- \phi(\overline{\by}_k)
\eeqan
for every $k$ in $\N$ and $\overline{\by}_k,\ldots,
\overline{\by}_0$ in $\mathbb{Y}$. We define $\Phi_R$ as 
the set of all stationary randomized policies for $\overline{\bY}$ 
that satisfy (\ref{ActionConstraintsauxiliary}).
\end{defn}

Henceforth, we use $\obY^{\phi}$ to denote the 
the auxiliary MDP $\obY$ under a policy
$\phi$ in $\Phi_R$.

Following the approach in~\cite{Lin2019Scheduling-task}, we restrict our analysis to the subset $\Phi^+_R$ of $\Phi_R$ defined as follows: 
\begin{equation*}
\Phi^+_R 
:= \{ \phi \in \Phi_R \ | \ \phi(1,\cA) > 0 \}
\end{equation*} 
The main benefit of focusing on policies in $\Phi^+_R$, 
as stated in \cite[Corollary~1]{Lin2019Scheduling-task}, 
is that $\obY^{\phi}$
has a unique stationary PMF (denoted with $\obpi^{\phi}$) for every 
$\phi$ in $\Phi^+_R$. Specifically, 
that strategies in $\Phi^+_R$ rule out the 
case in which $(1,\cA)$ is an absorbing state, guarantees 
the uniqueness of the stationary PMF. Furthermore, from 
\cite[Lemmas~2 and~4]{Lin2019Scheduling-task} we conclude 
that restricting to $\Phi^+_R$ any search that seeks to 
determine bounds or fundamental limits with respect to 
stabilizing policies incurs no loss of generality.


\subsection{Service Rate of $\overline{\mathbf{Y}}^{\phi}$ and Pr\'{e}cis of  Stabilizability Results}
We start by defining the service rate of 
$\overline{\mathbf{Y}}^{\phi}$ for a given policy 
$\phi$ in $\Phi_R^+$:
\beqan
\onu^{\phi} 
\Eqdef
	\sum_{\overline{\by} \in \mathbb{Y}} \mu(\overline{s})   \phi(\overline{\by})
	\overline{\pi}^{\phi}(\overline{\by}).
	\label{eq:onu} 
\eeqan
The maximal service rate $\onu^*$ for $\obY$ is defined below.
\beqan
\onu^* 
\Eqdef \sup_{\phi \in \Phi_R^+} \onu^{\phi}
\eeqan

As stated in~\cite[Theorems~3.1 and~3.2]{Lin2019Scheduling-task}, any arrival rate $\lambda$ lower than  $\onu^*$ is stabilizable. Furthermore, these theorems also assert that any arrival rate above $\onu^*$ is not stabilizable and that $\onu^*$ can also be computed by determining which threshold policy $\phi_{\tau}$, among the finitely many defined in~\cite[(6)]{Lin2019Scheduling-task}, maximizes $\onu^{\phi_{\tau}}$.

\begin{defn} We define the map $\mathscr{X}:\Phi_R^+ \rightarrow \Theta_R$ as follows:
\beqan
\mathscr{X}(\phi) \Eqdef \vartheta^{\phi}, 
	\quad \phi \in \Phi_R^+,
\eeqan 
where 
\beqa
\label{eq:structure_policies}
\vartheta^{\phi} (\bx) \Eqdef \begin{cases} \phi(\by) & \text{if $q > 0$} \\ 0 & \text{otherwise} \end{cases}, \quad \bx \in \mathbb{X}
\eeqa
\end{defn}

It follows from its definition that $\mathscr{X}$ yields a policy for $\bX$ that acts as the given $\phi$ in $\Phi_R^+$ when the queue is not empty and imposes rest otherwise. 

\paragraph*{\bf Convention} We reserve $\onu$, without a superscript, to denote a design parameter. Namely, it is a desired service rate that will be imposed as a constraint in the definition of the following policy sets.

\begin{defn}
\label{def:PolicySetPhis}
{\bf (Policy sets $\Phi^{\epsilon}_R(\onu)$ and $\Phi^+_R(\onu)$)} Given $\onu$ in $(0,\onu^*)$, we define the following policy sets:
\begin{align*}
\Phi^+_R(\onu) 
& \Eqdef \{ \phi \in \Phi_R^+ \ | \ 
	\onu^{\phi} = \onu \} \\
\Phi^{\epsilon}_R(\onu) 
& \Eqdef \{ \phi \in \Phi_R^{\epsilon} 
	\ | \ \onu^{\phi} = \onu \}, 
	\quad \epsilon \in [0,1] 
\end{align*} 
where $\Phi_R^{\epsilon}$ is defined as
\begin{equation*}
\Phi^{\epsilon}_R 
:= \{ \phi \in \Phi_R \ | \ \phi(1,\cA) \geq \epsilon \}, 
	\quad \epsilon \in [0,1]
\end{equation*}
\end{defn}

We  also define the following class of policies generated from $\Phi^+_R(\onu)$ and $\Phi^{\epsilon}_R(\onu)$ through $\mathscr{X}$:
\begin{multline*}
\myb \mathscr{X} \Phi_R^{\epsilon}(\onu) 
\Eqdef \{\mathscr{X}(\phi) :
	\phi \in \Phi^{\epsilon}_R(\onu) \}, \
	\onu \in (0,\onu^*), \ \ \epsilon \in (0,1]
\end{multline*}
\begin{equation*}
\mathscr{X} \Phi_R^+(\onu) 
\Eqdef \{\mathscr{X}(\phi) : 
	\phi \in \Phi^+_R(\onu) \}, 
	\ \  \onu \in (0,\onu^*)
\end{equation*}

The following proposition establishes important stabilization properties for the policies in $\mathscr{X} \Phi_R^+ (\onu)$.
\begin{prop}
\label{prop:SufficientCondStab} Let the arrival rate $\lambda$ in $(0,\onu^*)$ be given. If $\onu$ is in $(\lambda,\onu^*)$, then $\bX^{\theta}$ is stable, irreducible and aperiodic for any $\theta$ in $\mathscr{X} \Phi_R^+(\onu)$.
\end{prop}
\begin{proof} Stability of $\bX^{\theta}$ can be established using the same method adopted in~\cite{Lin2019Scheduling-task} to prove~\cite[Theorem~3.2]{Lin2019Scheduling-task}, which uses~\cite[Lemma~8]{Lin2019Scheduling-task} to establish a contradiction when $\bX^{\theta}$ is assumed not stable. That $\bX^{\theta}$ is irreducible follows from the fact that, under any policy $\theta$ in $\mathscr{X} \Phi_R^+(\onu)$, all states of $\bX^{\theta}$ communicate with $(1,\mathcal{A},0)$. That the probability of transitioning away from $(1,\mathcal{A},0)$ is less than one implies that the chain is aperiodic.
\end{proof}



An immediate consequence of Proposition~\ref{prop:SufficientCondStab} is that 
$\{ \mathscr{X}(\phi) : \phi \in \Phi_R^+(\onu) \}$ 
is a nonempty subset of $\Theta_S(\lambda)$ when 
$\lambda < \onu \leq \onu^*$. This implies that, 
as far as stabilizability is concerned, there is 
no loss of generality in restricting our analysis 
to policies with the structure in~(\ref{eq:structure_policies}). More interestingly, 
from Theorem~\ref{thm:PolicyDesign}, which will be 
stated and proved later on in Section
\ref{sec:MainResults}, we can conclude that 
restricting our methods for solving Problem~2 to 
policies of the form~(\ref{eq:structure_policies}) 
also incurs no loss of generality.


The following projection map will be important going forward.

\begin{defn}[Policy projection map $\mathscr{Y}$] 
Given $\lambda$ in $(0,\onu^*)$, we define a mapping
$\mathscr{Y}: \Theta_S(\lambda) \rightarrow \Phi_R^+$, where
\beqan
\mathscr{Y}(\theta)\Eqdef \varphi^{\theta}, \ \theta 
	\in \Theta_S(\lambda) 
\eeqan
with 
\beqan
\label{QMapDef}
\varphi^{\theta}(\overline{\by}) 
\Eqdef \frac{\sum_{q \in \mathbb{Q}^{\overline{\by}}}
	\theta(\overline{\by},q) \pi^{\theta}(\overline{\by},q)}
	{\sum_{q \in \mathbb{Q}^{\overline{\by}}} 
	\pi^{\theta}(\overline{\by},q)}, \quad
	\overline{\by} \in \mathbb{Y},
\eeqan 
where
$\mathbb{Q}^{\overline{\by}} \Eqdef \{q \in \N \ | \ 
	(\overline{\by},q)\in \mathbb{X} \}$, $\oby \in \mathbb{Y}$.
\end{defn}

Notice that although the map $\mathscr{Y}$ depends on $\lambda$, for simplicity of notation, we chose not to denote this
dependence explicitly. It is worthwhile to note that the map $\mathscr{Y}$, for a given $\lambda$ less than $\onu^*$, allows us to establish the following remark comparing the service rate notions for $\bX$ and $\obY$.

\begin{remark} Given $\lambda$ in $(0,\onu^*)$ and $\onu$ in $(\lambda,\onu^*)$, our analysis in~\cite{Lin2019Scheduling-task} implies that the following hold:
\begin{subequations}
\label{eqs:InequalitiesServiceRate}
\begin{align}
\label{eqs:InequalitiesServiceRate-a}
\lambda & \ \overset{(i)}{=} \ \nu^{\theta} \ \overset{(ii)}{=} \ \onu^{\mathscr{Y}(\theta)} \leq \onu^*, & \  \theta \in \Theta_S(\lambda) \\
\lambda & \ \overset{(iii)}{=} \ \nu^{\mathscr{X}(\phi)} \ < \ \onu \leq \onu^*, & \  \phi \in \Phi_R^+(\onu)
\end{align}
\end{subequations}
where $\nu^\theta := \sum_{\bx \in \mathbb{X}}
\pi^\theta(\bx) \theta(\bx) \mu(s)$ 
is the service rate of $\mathbf{X}^\theta$.
Notably, $(i)$ and $(ii)$ follow from~\cite[Lemma~4]{Lin2019Scheduling-task}. Using a similar argument, $(iii)$ follows from the fact that $\mathscr{X}(\phi)$ is stabilizing, as guaranteed by Proposition~\ref{prop:SufficientCondStab} when $\onu$ is in~$(\lambda,\onu^*)$.
\end{remark}




\subsection{Utilization Rate of $\overline{\bY}$ and 
    Computation via LP}

We now proceed to defining the utilization rate of $\obY^{\phi}$ for a given $\phi$ in $\Phi_R$. Subsequently, we will define and propose a linear programming approach to computing the infimum of the utilization rates attainable by any policy for $\obY$ subject to a given service rate.

\begin{defn} 
\label{defn:UtilizationRate}
Given a policy $\phi$ in $\Phi_R^+$, the following function 
determines the utilization rate of $\overline{\bY}^{\phi}$:
\begin{equation*}
\bar{\mathscr{U}}(\phi) := \sum_{\overline{\by} \in \overline{\mathbb{Y}}} \overline{\pi}^{\phi} \phi(\overline{\by})
\end{equation*}
\end{defn}


\begin{defn} {\bf (Infimum utilization rate 
$\bar{\mathscr{U}}^+_R$  and~$\bar{\mathscr{U}}^{\epsilon}_R$)} 
The infimum utilization rate of $\overline{\bY}$ for a given 
service rate $\overline{\nu}$ is defined as
\begin{equation*}
\label{eq:DefInfimumRateInY}
\bar{\mathscr{U}}^+_R (\overline{\nu}) 
:= \inf_{\phi\in \Phi_R^+(\onu)} \quad 
	\sum_{\overline{\mathbf{y}} 
		\in \overline{\mathbb{Y}}} 
		\overline{\pi}^\phi(\overline{\mathbf{y}})
		\phi(\overline{\mathbf{y}}).
\end{equation*} 
We also define the following approximate infimum utilization 
rates:
\begin{equation*}
\label{eq:DefInfimumRateInYForEpsilon}
\bar{\mathscr{U}}^{\epsilon}_R (\overline{\nu}) 
:= \inf_{\phi\in \Phi_R^{\epsilon}(\onu)} \quad 
	\sum_{\overline{\mathbf{y}}
		\in \overline{\mathbb{Y}}} 
		\overline{\pi}^\phi(\overline{\mathbf{y}})
		\phi(\overline{\mathbf{y}}), 
		\quad \epsilon \in (0, 1]
\end{equation*}
\end{defn}

Notice that the infimum that determines $\bar{\mathscr{U}}^+_R$ and $\bar{\mathscr{U}}^{\epsilon}_R$ is well-defined because there is a unique stationary PMF $\overline{\pi}^\phi$ for each policy $\phi$ in $\Phi_R^+$. 

\begin{remark} 
\label{rem:ContinuityInEpsilonUbar} Notice that since $\Phi_R^+(\onu) = \bigcup_{\epsilon \in (0,1]} \Phi_R^{\epsilon}(\onu) $, we conclude that the following holds:
\begin{equation}
\label{eq:LimEpsilonZeroUBarR}
\bar{\mathscr{U}}^+_R (\overline{\nu}) = \lim_{\epsilon \rightarrow 0^+} \bar{\mathscr{U}}^{\epsilon}_R (\overline{\nu}) 
\end{equation}
\end{remark} 

We now proceed to outlining efficient ways to compute $\bar{\mathscr{U}}^+_R$, which is relevant because, as Corollary~\ref{cor:Equality} indicates in \S\ref{sec:MainResults}, we can use it to compute $\mathscr{U}^* (\lambda)$ when $\lambda < \onu^*$.  Hence, below we follow the approach in~\cite[Chapter~4]{Altman1999Constrained-Mar} to construct approximate versions of $\bar{\mathscr{U}}^+_R$ that are computable using a finite-dimensional LP. Subsequently, we will obtain the policies in~$\Phi_R^+$ corresponding to solutions of the LP, as is done in~\cite[Chapter~4]{Altman1999Constrained-Mar}. The policies obtained in this way will form a set for each $\epsilon$ in $(0,1)$ that will be useful later on.

\begin{defn}{\bf ($\epsilon$-LP utilization rate $\bar{\mathscr{U}}^{\epsilon}_{\mathbb{L}}(\overline{\nu})$)} \\ 
Let $\epsilon$ be a given constant in $[0,1]$ and $\overline{\nu}$ be a pre-selected service rate in $[0,\onu^*]$. The $\epsilon$-LP utilization rate $\bar{\mathscr{U}}^{\epsilon}_{\mathbb{L}}(\overline{\nu})$ is defined as:
\begin{subequations}
\label{LP-Definition}
\begin{equation}
\label{eq:minimization}
\bar{\mathscr{U}}^{\epsilon}_{\mathbb{L}}(\overline{\nu}) : = \min_{ \begin{matrix} \ell \in \mathbb{L} \\ \text{s.t. (\ref{LP-constraint-b})-(\ref{LP-constraint-e})} \end{matrix}} \quad \sum_{\overline{\mathbf{y}}\in\overline{\mathbb{Y}}} \ell_{\overline{\mathbf{y}},\mathcal{W}}
\end{equation} 

\noindent where the minimization is carried out over the following set:
\begin{equation*}
\mathbb{L} := \Pi_{ \overline{a} \in \overline{\mathbb{A}}_{\overline{\by}}, \overline{\by} \in \mathbb{Y} } \{\ell_{\overline{\by},\overline{a}} \geq 0 \}
\end{equation*} Every solution is subject to the following constraints and is compactly represented as $\ell:=\Pi_{ \overline{a} \in \overline{\mathbb{A}}_{\overline{\by}}, \overline{\by} \in \mathbb{Y} } \{ \ell_{\overline{\by},\overline{a}} \}$:

\begin{align}
\label{LP-constraint-b}
(1-\epsilon) \ell_{(1,\mathcal{A}),\mathcal{W}}& 
\geq \epsilon \ell_{(1,\mathcal{A}),\mathcal{R}} \\ \label{LP-constraint-c} 
\sum_{ \{ \overline{\mathbf{y}} \in
	\overline{\mathbb{Y}} | \mathcal{W} \in 
	\overline{\mathbb{A}}_{\overline{\mathbf{y}} } \} }  
	\mu(\overline{s}) 
	\ell_{\overline{\mathbf{y}},\mathcal{W}} 
	&  = \overline{\nu} \\ 
\label{LP-constraint-d}
\sum_{\overline{\mathbf{y}} \in \overline{\mathbb{Y}}}
	\sum_{\overline{a} \in 
	\overline{\mathbb{A}}_{\overline{\mathbf{y}}}} 
	\ell_{\overline{\mathbf{y}},\overline{a}} &=1
\end{align}
and the equality below guarantees that every solution will be consistent with $\overline{\bY}$:
\begin{eqnarray}
&& \sum_{\overline{\mathbf{y}}' \in \overline{\mathbb{Y}}}
	\sum_{\overline{a}' \in 
	\overline{\mathbb{A}}_{\overline{\mathbf{y}}'}} 
	\ell_{\overline{\mathbf{y}}',\overline{a}'} \ 
	P_{\overline{\mathbf{Y}}_{t+1} | 
	    \overline{\mathbf{Y}}_{t}, 
	    \overline{A}_t} \big(\overline{\mathbf{y}} 
	    \ \big| \ \overline{\mathbf{y}}', \overline{a}' \big) 
	\nonumber \\  
&& = \sum_{\overline{a} \in
	\overline{\mathbb{A}}_{\overline{\mathbf{y}}}} 
	\ell_{\overline{\mathbf{y}},\overline{a}}, 
	\quad \overline{\mathbf{y}} \in 
		\overline{\mathbb{Y}}
\label{LP-constraint-e}
\end{eqnarray}
\end{subequations}
\end{defn}

\vspace{.1 in}

\begin{defn} {\bf (Solution set $\mathbb{L}^{\epsilon}(\overline{\nu})$)}  For each $\epsilon$ in $[0,1]$ and $\overline{\nu}$ in $(0,\lambda^*)$, we use~$\mathbb{L}^{\epsilon}(\overline{\nu})$ to represent the set of solutions of the LP specified by (\ref{LP-Definition}). We adopt the convention that $\mathbb{L}^{\epsilon}(\overline{\nu})$ is empty if and only if the LP is not feasible.
\end{defn}

\begin{figure} [ht]
\centering
\begin{tikzpicture}[scale=.072]


\draw[thin, fill=black!3, rounded corners=5] (10,8) rectangle (104,35); 
\node[anchor= north west] at (10,35) {\small $\Theta_R$};

\draw[thin, fill=black!4, rounded corners=5] (20,9) rectangle (103,33); 
\node[anchor= north west] at (20,33) {\small $\Theta_S(\lambda)$};

\draw[thin, fill=black!5, rounded corners=5] (35,10) rectangle (102,31); 
\node[anchor= north west] at (35,31) {\small $\mathscr{X} \Phi_{R}^+ (\onu)$};

\draw[thin, fill=black!6, rounded corners=5] (55,11) rectangle (101,29); 
\node[anchor= north west] at (55,29) { \small $\mathscr{X} \Phi_{R}^{\epsilon} (\onu)$};

\draw[thin, fill=black!7, rounded corners=5] (75,12) rectangle (100,27); 
\node[anchor= north west] at (75,27) {\small $\mathscr{X} \Phi_{\mathbb{L}}^{\epsilon} (\onu)$};


\draw[thin, fill=black!3, rounded corners=5] (10,-9) rectangle (104,-35); 
\node[anchor= south west] at (10,-35) {\small $\Phi_R^+$};

\draw[thin, fill=black!5, rounded corners=5] (35,-10) rectangle (102,-31); 
\node[anchor= south west] at (35,-31) {\small $ \Phi_{R}^+ (\onu)$};

\draw[thin, fill=black!6, rounded corners=5] (55,-11) rectangle (101,-29); 
\node[anchor= south west] at (55,-29) {\small $\Phi_{R}^{\epsilon} (\onu)$};

\draw[thin, fill=black!7, rounded corners=5] (75,-12) rectangle (100,-27); 
\node[anchor= south west] at (75,-27) { \small $\Phi_{\mathbb{L}}^{\epsilon} (\onu)$};

\draw[thick,->] (87,-17) -- (87,17);
\node[anchor=west] at (87,0) {\small $\mathscr{X}$};

\draw[thick,->] (65,-16) -- (65,16);
\node[anchor=west] at (65,0) {\small $\mathscr{X}$};

\draw[thick,->] (45,-15) -- (45,15);
\node[anchor=west] at (45,0) {\small $\mathscr{X}$};

\draw[thick,->] (28,14) -- (28,-14);
\node[anchor=west] at (28,0) {\small $\mathscr{Y}$};

\draw[thick,->] (16,-13) -- (16,13);
\node[anchor=west] at (16,0) {\small $\mathscr{X}$};

\node[anchor=south] at (52.5,37) {\small policy sets for $\bX$};

\node[anchor=north] at (52.5,-37) {\small policy sets for $\obY$};

\end{tikzpicture}
\caption{Diagram representing the relationship among policy sets for $\bX$ and $\obY$, for $0 < \lambda < \nu < \onu < \onu^*$ and $\epsilon \in (0,1]$. 
}
\label{Fig:VennDiag}
\end{figure}

\subsection{LP-based Policy Sets}

For each solution $\ell$ in $\mathbb{L}^{\epsilon}(\onu)$ we can obtain a corresponding policy $\varphi_{\ell}$ in $\Phi_R$ for $\obY$ as follows:
\begin{equation}
\label{eq:def_policy_from_LP}
\varphi_{\ell}(\oby) 
:=
\begin{cases}
\frac{\ell_{\oby,\mathcal{W}}}{\ell_{\oby,\mathcal{W}} + \ell_{\oby,\mathcal{R}}} & \text{if $\mathcal{R} \in \overline{\mathbb{A}}_{\oby}$ and $\ell_{\oby,\mathcal{R}} >0$} \\
1 & \text{otherwise.}
\end{cases} , \  \oby \in \overline{\mathbb{Y}}
\end{equation}

\begin{remark} 
\label{rem:OnTheConstraints} Subject to the definition in~(\ref{eq:def_policy_from_LP}), the constraint~(\ref{LP-constraint-b}) is equivalent to $\varphi_{\ell}(1,\mathcal{A}) \geq \epsilon$, which holds for every solution $\ell$ in $\mathbb{L}^{\epsilon}(\overline{\nu})$. 
\vspace{.1 in}

\end{remark}

\begin{defn}{\bf (Policy set $\Phi_{\mathbb{L}}^{\epsilon}(\overline{\nu})$)}  For each $\overline{\nu}$ in $(0,\onu^*)$ and $\epsilon$ in $[0,1]$, we define the following set of policies $\Phi_{\mathbb{L}}^{\epsilon}(\overline{\nu})$: 
\begin{equation*}
\Phi_{\mathbb{L}}^{\epsilon}(\overline{\nu})
:= \{ \varphi_{\ell} : \ell \in 
	\mathbb{L}^{\epsilon}(\overline{\nu}) \}
\end{equation*} 
Here, we adopt the convention that~$\Phi_{\mathbb{L}}^{\epsilon} (\overline{\nu})$ is empty if and only if~$\mathbb{L}^{\epsilon}(\overline{\nu})$ is empty.
\vspace{.05 in}
\end{defn}

The following proposition will justify choices for $\epsilon$ we will make at a later stage to guarantee that $\Phi_{\mathbb{L}}^{\epsilon}(\onu)$ is nonempty for $\onu$ in $(\lambda,\onu^*)$.

\begin{prop}
\label{prop:PhiEpsilonNonEmpty} Suppose that $\onu^-$ lies in 
$(0, \onu^*]$ and there is $\epsilon^*$ in $(0,1]$ such that 
$\mathbb{L}^{\epsilon^*}(\onu^-)$ is nonempty. Then, 
$\mathbb{L}^{\bar{\epsilon}}(\onu)$ is nonempty for any 
$\bar{\epsilon}$ in $(0,\epsilon^*]$ and $\onu$ in $[\onu^-,
\onu^*]$. 
\end{prop}

\begin{proof} 
We start by invoking~\cite[Lemma~7]{Lin2019Scheduling-task} 
to conclude that $\Phi_{\mathbb{L}}^1(\onu^*)$ is nonempty 
and, consequently, $\mathbb{L}^{\epsilon^*}(\onu^*)$ 
is also nonempty. Suppose that $\ell^{\onu^-}$ and $\ell^{\onu^*}$ 
are in $\mathbb{L}^{\epsilon^*}(\onu^-)$ and 
$\mathbb{L}^{\epsilon^*}(\onu^*)$, respectively. Then, from 
(\ref{LP-Definition}), for any $\onu$ in 
$[\onu^-, \onu^*]$, $\ell^{\onu}
\Eqdef \big( (\onu-\onu^-) \ell^{\onu^*}
+ (\onu^*-\onu) \ell^{\onu^-} \big) / (\onu^*-\onu^-)$ 
satisfies 
(\ref{LP-constraint-b})-(\ref{LP-constraint-e}), which 
implies that $\mathbb{L}^{\epsilon^*}(\onu)$ is nonempty.
That $\mathbb{L}^{\epsilon^*}(\onu)$ is nonempty implies that $\mathbb{L}^{\bar{\epsilon}}(\onu)$ is also nonempty for any $\bar{\epsilon}$ in $(0,\epsilon^*]$, which concludes the proof.
\end{proof}

Before we proceed with stating a proposition that has important implications for design, we define the following notion of dominance also used in~\cite{Altman1999Constrained-Mar}.

\begin{defn} {\bf (Policy set dominance)} Let $\onu$ in $(0,\onu^*)$ and any two subsets $\tilde{\Phi}_1$ and $\tilde{\Phi}_2$ of $\Phi_R^+(\onu)$ be given. We say that $\tilde{\Phi}_1$ dominates $\tilde{\Phi}_2$ if for each policy $\phi_2$ in $\tilde{\Phi}_2$ there is $\phi_1$ in $\tilde{\Phi}_1$ for which $\bar{\mathscr{U}}(\phi_1) \leq \bar{\mathscr{U}}(\phi_2)$.
\end{defn}

\vspace{0.1 in} 
\begin{prop} 
\label{prop:CanUseLP} Given $\overline{\nu}$ in $(0,\onu^*)$ and $\epsilon$ in $(0,1]$, $\Phi^{\epsilon}_{\mathbb{L}} (\overline{\nu})$ dominates $\Phi_R^{\epsilon}(\onu)$ and the equality below holds:
\begin{subequations}
\begin{equation} \label{eq:UepsilonWorks}
 \bar{\mathscr{U}}^{\epsilon}_R(\overline{\nu}) = \bar{\mathscr{U}}^{\epsilon}_{\mathbb{L}}(\overline{\nu})
\end{equation}
\begin{equation} \label{eq:UzeroWorks}
\bar{\mathscr{U}}^+_R(\overline{\nu}) = \bar{\mathscr{U}}^{0}_{\mathbb{L}}(\overline{\nu})
\end{equation}
\end{subequations}
\end{prop}

\begin{proof}
It follows immediately from~\cite[Theorem~4.3]{Altman1999Constrained-Mar} that~(\ref{eq:UepsilonWorks}) holds and~$\Phi^{\epsilon}_{\mathbb{L}} (\overline{\nu})$ dominates $\Phi_R^{\epsilon}(\onu)$. 
Furthermore, Proposition~\ref{prop:U-EpsilonCont} from \S\ref{sec:ContinuityProperties} implies that the following limit holds:

\begin{equation}
\label{eq:SequenceLimitEpsilon}
\lim_{\epsilon \rightarrow 0^+} \ \bar{\mathscr{U}}^{\epsilon}_{\mathbb{L}}(\overline{\nu}) = \bar{\mathscr{U}}^{0}_{\mathbb{L}}(\overline{\nu})
\end{equation}
That~(\ref{eq:UzeroWorks}) holds is a consequence of~(\ref{eq:LimEpsilonZeroUBarR}), (\ref{eq:UepsilonWorks}) and  (\ref{eq:SequenceLimitEpsilon}).
\end{proof}

Before proceedings to describing our main results, we define the following class of policies for $\bX$ that can be generated from solutions of the~LP~(\ref{LP-Definition}):
\begin{equation*}
\mathscr{X} \Phi_{\mathbb{L}}^{\epsilon} (\onu) \Eqdef \{\mathscr{X}(\phi) : \phi \in \Phi^{\epsilon}_{\mathbb{L}}(\onu) \}, \ \onu \in (0,\onu^*), \ \epsilon \in (0,1]
\end{equation*}

See Fig.\ref{Fig:VennDiag} for a representation of the relationships among most of the policy sets for the MDPs $\bX$ and $\obY$.

\section{Main Results}
\label{sec:MainResults}

This section starts with Theorem~\ref{thm:PolicyDesign}, which is our main result. Subsequently, we state Corollaries~\ref{cor:Equality} and~\ref{cor:PolicyDesign} that undergird our methods to tackle Problems~\ref{Problem:ComputeBound} and~\ref{Problem:ComputePolicy}, respectively. 

\begin{theorem}
\label{thm:PolicyDesign} Let an arrival rate $\lambda$ in $(0,\onu^*)$ be given. For each positive gap $\delta$,
there exist a service rate $\onu^{\delta,\lambda}$ in $(\lambda,\onu^*)$ and $\epsilon^{\delta,\lambda}$ in $(0,1]$ such that $\Phi^{\epsilon^{\delta,\lambda}}_{\mathbb{L}}(\onu^{\delta,\lambda})$ is nonempty and the following inequality holds:
\beqa
\label{eq:IneqMainThm}
\mathscr{U}(\lambda,\theta) \leq
\bar{\mathscr{U}}^+_R(\lambda) + \delta, \quad \theta \in \mathscr{X} \Phi^{\epsilon^{\delta,\lambda}}_{\mathbb{L}}(\onu^{\delta,\lambda})
\eeqa
\end{theorem}

Remarks~\ref{rem:EqualityComputation} and~\ref{rem:Cor2} will expound the importance of Theorem~\ref{thm:PolicyDesign} and its two corollaries.
Our proof of the theorem given below relies on the continuity properties and distributional convergence results established in \S\ref{sec:ContinuityProperties} and~\S\ref{sec:DistributionalConvergence}, respectively.

\begin{proof}
Since it follows from Theorem~\ref{thm:U0Convex} 
in \S\ref{sec:ContinuityProperties} that 
$\bar{\mathscr{U}}^0_{\mathbb{L}}$ is continuous 
and non-decreasing, we know that there is 
$\onu^{\dagger}$ in $(\lambda,\onu^*)$ such that 
the following inequality holds:
\begin{equation}
\label{eq:ProofMainTh-1}
  \bar{\mathscr{U}}^0_{\mathbb{L}}(\onu^{\dagger}) \leq  \bar{\mathscr{U}}^0_{\mathbb{L}}(\lambda) + \tfrac{1}{3} \delta
\end{equation}

Since $\lambda$ is stabilizable, a stabilizing policy $\theta\in\Theta_S(\lambda)$ exists. By \cite[Lemma 2]{Lin2019Scheduling-task}, $\mathscr{Y}(\theta)$ has non-zero probability to choose to work at state $(1,\cA)$ and $\onu^{\mathscr{Y}(\theta)}=\lambda$ by~\eqref{eqs:InequalitiesServiceRate-a}. Therefore, $\Phi_R^{\epsilon^{\dagger}}(\lambda)$ is nonempty for some positive $\epsilon^\dagger$. From Proposition~\ref{prop:U-EpsilonCont} in~\S\ref{sec:ContinuityProperties}, we can select $\epsilon^{\ddagger}$ in $(0,\epsilon^{\dagger}]$ such that the following holds:
\begin{equation}
\label{eq:ProofMainTh-2}
    \bar{\mathscr{U}}_{\mathbb{L}}^{\epsilon}(\onu^{\dagger}) \leq \bar{\mathscr{U}}_{\mathbb{L}}^0(\onu^{\dagger}) + \tfrac{1}{3} \delta, \quad \epsilon\in(0,\epsilon^{\ddagger}]
\end{equation}
From Proposition~\ref{prop:U-OnuNonIncrease} in~\S\ref{sec:ContinuityProperties} we know that we can choose $\epsilon^{\delta,\lambda}$ in $(0,\epsilon^{\ddagger}]$ such that the following holds: 
\begin{equation}
\label{eq:ProofMainTh-3}
\bar{\mathscr{U}}^{\epsilon^{\delta,\lambda}}_{\mathbb{L}}(\onu) \leq \bar{\mathscr{U}}^{\epsilon^{\delta,\lambda}}_{\mathbb{L}}(\onu^\dagger) , \quad\onu\in(\lambda, \onu^\dagger)
\end{equation}

In \S\ref{sec:DistributionalConvergence} we develop in sequence several results that ultimately lead to Theorem~\ref{thm:MainDistConv}, which establishes an important distributional convergence result that takes hold when $\onu$ in $(\lambda,\onu^{\dagger})$ is selected as close as needed to $\lambda$. Using Corollary~\ref{coro:MainUDistConverence} stated also in~\S\ref{sec:DistributionalConvergence}, which follows immediately from Theorem~\ref{thm:MainDistConv}, we conclude that, based on our choice of $\epsilon^{\delta,\lambda}$ above, we can select $\onu^{\delta,\lambda}$ in $(\lambda,\onu^{\dagger})$ such that the following inequality holds:
\begin{equation}
\label{eq:ProofMainTh-4}
    \mathscr{U}\big ( \lambda,\mathscr{X}(\phi) \big ) \leq \bar{\mathscr{U}}(\phi) + \tfrac{1}{3} \delta , \quad \phi \in \Phi^{\epsilon^{\delta,\lambda}}_{\mathbb{L}} (\onu^{\delta,\lambda})
\end{equation}

Hence, using our choices for $\epsilon^{\delta,\lambda}$ and $\onu^{\delta,\lambda}$, we infer from~(\ref{eq:ProofMainTh-1})-(\ref{eq:ProofMainTh-4}) that the following inequality holds:
\begin{equation}
    \mathscr{U}\big ( \lambda,\mathscr{X}(\phi) \big ) \leq \bar{\mathscr{U}}^0_{\mathbb{L}}(\lambda) + \delta, \quad \phi \in \Phi^{\epsilon^{\delta,\lambda}}_{\mathbb{L}} (\onu^{\delta,\lambda})
\end{equation} which, together with~(\ref{eq:UzeroWorks}), leads to (\ref{eq:IneqMainThm}).
\end{proof}

We proceed with stating a proposition that provides an utilization-rate counterpart for~(ii) in~(\ref{eqs:InequalitiesServiceRate-a}) and whose proof we omit because it follows immediately from~\cite[Lemmas~3 and~4]{Lin2019Scheduling-task}.

\begin{prop}
\label{prop:ProjectionPreservesUtilization} Given $\lambda$ in $(0,\onu^*)$, the following equality holds for any $\theta$ in $\Theta_S(\lambda)$:
\begin{equation}
\bar{\mathscr{U}} \big( \mathscr{Y}(\theta) \big) = \mathscr{U}(\lambda,\theta) 
\end{equation}
\end{prop}


\begin{coro}
\label{cor:Equality} The following equality holds:
\beqa
\label{eq:coroEquality}
\mathscr{U}^*(\lambda) = \bar{\mathscr{U}}^+_R(\lambda), \quad \lambda \in (0,\onu^*) 
\eeqa
\end{coro}
\begin{proof}
It ensues from Proposition~\ref{prop:ProjectionPreservesUtilization} and (i)-(ii) in~(\ref{eqs:InequalitiesServiceRate-a}) that the following holds for any $\lambda$ in $(0,\onu^*)$:
\beqa
\mathscr{U}(\lambda, \theta) = \bar{\mathscr{U}}\big( \mathscr{Y}(\theta) \big) \geq \bar{\mathscr{U}}^+_R(\lambda), \quad \theta \in \Theta_S(\lambda)
\eeqa
Since the inequality above holds for any $\theta$ in $\Theta_S(\lambda)$, the following inequality is satisfied for any $\lambda$ in $(0,\onu^*)$:
\beqa
\label{eq:coro1proof}
\mathscr{U}^*(\lambda) \geq \bar{\mathscr{U}}^+_R(\lambda)
\eeqa

We conclude the proof by remarking that~(\ref{eq:coro1proof}) and Theorem~\ref{thm:PolicyDesign} imply (\ref{eq:coroEquality}).
\end{proof}

\begin{remark}[Solution of Problem~\ref{Problem:ComputeBound}] \label{rem:EqualityComputation} 
Corollary~\ref{cor:Equality} is significant because, in conjunction with (\ref{eq:UzeroWorks}), it indicates that $\mathscr{U}^*(\lambda)$ can be computed using the finite dimensional LP~(\ref{LP-Definition}) for $\epsilon=0$ and $\onu=\lambda$.
\end{remark}

Section~\ref{subsec:GraphicalExample} discusses a numerical example and a graphical method to determine $\bar{\mathscr{U}}^{0}_{\mathbb{L}}(\onu)$ for all values of $\onu$ in $[0,\onu^*]$. The graphical method leverages the analysis in Appendix~\ref{app:StructuralResults}, which establishes that $\bar{\mathscr{U}}^{0}_{\mathbb{L}}$ is non-decreasing and convex.

The following corollary follows directly from Theorem~\ref{thm:PolicyDesign} and Corollary~\ref{cor:Equality}. 

\begin{coro}	\label{cor:PolicyDesign} 
Let an arrival rate $\lambda$ in $(0,\onu^*)$ be given. For each positive gap $\delta$ there exist a service rate $\onu^{\delta,\lambda}$ in $(\lambda,\onu^*)$ and $\epsilon^{\delta,\lambda}$ in $(0,1]$ such that $\Phi^{\epsilon^{\delta,\lambda}}_{\mathbb{L}}(\onu^{\delta,\lambda})$ is nonempty and the following inequality holds:
\beqan
\label{eq:IneqMainCor}
\mathscr{U}(\lambda,\theta) \leq \mathscr{U}^*(\lambda) + \delta, \quad \theta \in \mathscr{X} \Phi^{\epsilon^{\delta,\lambda}}_{\mathbb{L}}(\onu^{\delta,\lambda})
\eeqan
\end{coro}

\begin{remark}[Solution to Problem~\ref{Problem:ComputePolicy}] \label{rem:Cor2}
While, as explained in Remark~\ref{rem:EqualityComputation}, $\mathscr{U}^*(\lambda)$ can be computed effectively for any stabilizable $\lambda$, Corollary~\ref{cor:PolicyDesign} ascertains that we can address Problem~\ref{Problem:ComputePolicy}. Specifically, given a stabilizable $\lambda$ and any positive gap $\delta$, Corollary~\ref{cor:PolicyDesign} guarantees that we can find $\onu$ and $\epsilon$ such that any policy $\theta$ in $\mathscr{X} \Phi^{\epsilon}_{\mathbb{L}}(\onu)$ is not only stabilizing but the utilization rate of $\bX^{\theta}$ does not exceed $\mathscr{U}^*(\lambda)+\delta$. The proof of Theorem~\ref{thm:PolicyDesign} outlines a method for selecting such $\onu$ and $\epsilon$. This is a significant result because any solution of the LP~(\ref{LP-Definition}) can be used to obtain a policy in $\mathscr{X} \Phi^{\epsilon}_{\mathbb{L}}(\onu)$.
\end{remark}


\section{Continuity and monotonicity of~$\bar{\mathscr{U}}^{\epsilon}_{\mathbb{L}}$ }
\label{sec:ContinuityProperties}

We proceed with establishing three properties of $\bar{\mathscr{U}}^{\epsilon}_{\mathbb{L}}$ that are needed in the proof of our main results in~\S\ref{sec:MainResults}.

The following proposition establishes that when, 
for a given $\onu$ in $(0,\onu^*)$, 
$\bar{\mathscr{U}}^{\epsilon}(\onu)$ is viewed 
as a function of $\epsilon$, it is right 
continuous at $0$.

\vspace{.1 in}

\begin{prop}
\label{prop:U-EpsilonCont}
Let $\onu$ in $(0,\onu^*)$ be given. For any positive $\delta$, there is $\epsilon$ such that 
$\bar{\mathscr{U}}^{\epsilon}_{\mathbb{L}}(\onu) 
\leq \bar{\mathscr{U}}^0_{\mathbb{L}}(\onu) 
	+ \delta$.
\end{prop}
\begin{proof}
The statement of the proposition is false if and only if
there exists some $\onu$ in $(0, \onu^*)$ for which 
$\underbar{d} 
:= \lim_{\epsilon \rightarrow 0^+} 
	\bar{\mathscr{U}}^{\epsilon}_{\mathbb{L}}(\onu) 
	- \bar{\mathscr{U}}^0_{\mathbb{L}}(\onu) 
	> 0$.
We proceed to proving the proposition by contradiction 
by showing that the inequality above does not hold. 
Take $\epsilon$ positive such that 
$d:= \bar{\mathscr{U}}^{\epsilon}_{\mathbb{L}}(\onu) 
- \bar{\mathscr{U}}^0_{\mathbb{L}}(\onu)$ is in 
$[\underbar{d},2\underbar{d})$. Select $\ell^{\epsilon}$ 
and $\ell^{0}$ in $\mathbb{L}^{\epsilon}(\onu)$ and 
$\mathbb{L}^{0}(\onu)$, respectively. Define $\ell^{av} 
:=  \tfrac{1}{3} (\ell^{\epsilon} +  2 \ell^0)$, which 
satisfies~(\ref{LP-constraint-c})-(\ref{LP-constraint-e}). 
Given that $\epsilon$ is positive, $\ell^{av}$ will also 
satisfy (\ref{LP-constraint-b}) for some positive 
$\epsilon^*$, which implies that
$\bar{\mathscr{U}}^{\epsilon^*}_{\mathbb{L}}(\onu) 
- \bar{\mathscr{U}}^0_{\mathbb{L}}(\onu) \leq 
\tfrac{1}{3} d \leq \tfrac{2}{3} \underbar{d}$.
\end{proof}

The following proposition establishes a useful monotonicity property in terms of $\onu$.
\vspace{.1 in}

\begin{prop}
\label{prop:U-OnuNonIncrease}
Let $\onu^\dagger$ and $\onu^\ddagger$ in $(0,\onu^*)$ be given with $\onu^\dagger < \onu^\ddagger$. There exists a positive $\epsilon^*$ such that the following holds: 
\begin{equation*}
\bar{\mathscr{U}}^{\epsilon}_{\mathbb{L}}(\onu) 
\leq \bar{\mathscr{U}}^{\epsilon}_{\mathbb{L}}(\onu^\ddagger), 
	\quad\onu\in(\onu^\dagger, \onu^\ddagger),\  
	\epsilon\in(0,\epsilon^*]
\end{equation*}
\end{prop}
\begin{proof}
From \eqref{eq:minimization}, \eqref{LP-constraint-c}, and the fact that $\min_{s\in\mathbb{S}}\mu(s)$ is positive, we get
\begin{equation}
\label{eq:UUpperBound}
\bar{\mathscr{U}}^{\epsilon}_{\mathbb{L}}(\onu) \leq \frac{1}{\min_{s\in\mathbb{S}}\mu(s)} \onu, \quad \onu\in(0, \onu^*),\ 
\epsilon\in[0,1]
\end{equation}
We can find $\onu^{-}$ in $(0,\onu^\dagger)$ such that the following inequality holds:
\begin{equation}
\label{eq:UUpperBound2}
\frac{1}{\min_{s\in\mathbb{S}}\mu(s)} \onu^- \leq \bar{\mathscr{U}}^{0}_{\mathbb{L}}(\onu^\ddagger)
\end{equation}


Since $\onu^-$ is stabilizable, a stabilizing policy $\theta\in\Theta_S(\onu^-)$ exists. By \cite[Lemma 2]{Lin2019Scheduling-task}, $\phi^\theta:=\mathscr{Y}(\theta)$ has non-zero probability to choose to work at state $(1,\cA)$ and $\onu^{\phi^\theta}=\onu^-$ by~\eqref{eqs:InequalitiesServiceRate-a}. We construct an $\ell^{\onu^-}$ by the following definitions:
\beqan
\ell^{\onu^-}_{\oby,\mathcal{W}} &\Eqdef& \opi^{\phi^\theta}(\oby)\phi^\theta(\oby)\\
\ell^{\onu^-}_{\oby,\mathcal{R}} &\Eqdef& \opi^{\phi^\theta}(\oby)\big(1-\phi^\theta(\oby)\big) \quad \oby \in \overline{\mathbb{Y}}
\eeqan
It is clear that $\ell^{\onu^-}$ satisfies \eqref{LP-constraint-b} with some positive $\epsilon^*$ since $\phi^\theta(1,\cA) > 0$ and all other constraints in (\ref{LP-Definition}).
Therefore, $\mathbb{L}^{\epsilon^*}(\onu^-)$ is nonempty for some positive $\epsilon^*$. Consequently, we can further invoke Proposition~\ref{prop:PhiEpsilonNonEmpty} to infer that $\mathbb{L}^{\epsilon}(\onu^-)$ and $\mathbb{L}^{\epsilon}(\onu^\ddagger)$ are nonempty for every $\epsilon$ in $[0,\epsilon^*]$. Now, let $\epsilon$ be an arbitrary selection in $(0,\epsilon^*]$ and $\ell^{\onu^-}$ 
and $\ell^{\onu^\ddagger}$ be elements of $\mathbb{L}^{\epsilon}(\onu^-)$ and $\mathbb{L}^{\epsilon}(\onu^\ddagger)$, respectively. From (\ref{LP-Definition}) we conclude that, for any $\onu$ 
in $(\onu^\dagger, \onu^\ddagger)$, $\ell^{\onu} 
\Eqdef \big( (\onu-\onu^-) \ell^{\onu^\ddagger} + 
(\onu^\ddagger-\onu) \ell^{\onu^-} \big) / (\onu^\ddagger-\onu^-)$ 
satisfies (\ref{LP-constraint-b})-(\ref{LP-constraint-e}).
From \eqref{eq:minimization} and the definition of $\mathbb{L}^{\epsilon}(\onu^-)$ and $\mathbb{L}^{\epsilon}(\onu^\ddagger)$, we use $\ell^{\onu}\in\mathbb{L}^{\epsilon}(\onu)$ and obtain the following inequality:
\beqan
\bar{\mathscr{U}}^{\epsilon}_{\mathbb{L}}(\onu) \myleq \frac{\onu-\onu^-}{\onu^\ddagger-\onu^-}\bar{\mathscr{U}}^{\epsilon}_{\mathbb{L}}(\onu^\ddagger) + \frac{\onu^\ddagger-\onu}{\onu^\ddagger-\onu^-}\bar{\mathscr{U}}^{\epsilon}_{\mathbb{L}}(\onu^-)
\eeqan
Furthermore, from \eqref{eq:UUpperBound} and \eqref{eq:UUpperBound2}, the following inequalities hold, which completes the proof:
\beqan
\bar{\mathscr{U}}^{\epsilon}_{\mathbb{L}}(\onu) 
\myleq \frac{\onu-\onu^-}{\onu^\ddagger-\onu^-}\bar{\mathscr{U}}^{\epsilon}_{\mathbb{L}}(\onu^\ddagger) + \frac{\onu^\ddagger-\onu}{\onu^\ddagger-\onu^-}\bar{\mathscr{U}}^{0}_{\mathbb{L}}(\onu^\ddagger) 
\leq \bar{\mathscr{U}}^{\epsilon}_{\mathbb{L}}(\onu^\ddagger).
\eeqan
\end{proof}

The following theorem establishes important structural properties for $\bar{\mathscr{U}}^0_{\mathbb{L}}$. We provide a proof of the theorem in Appendix~\ref{app:StructuralResults}.

\begin{theorem}
\label{thm:U0Convex}
The $0$-LP utilization rate function ${\bar{\mathscr{U}}^0_{\mathbb{L}} : [0,\onu^*] \rightarrow [0,1]}$ is non-decreasing, piecewise affine and convex.
\end{theorem}

\subsection{A Graphical Method and Numerical Example}
\label{subsec:GraphicalExample}

 We proceed to describe a method to obtain $\bar{\mathscr{U}}^{0}_{\mathbb{L}}(\overline{\nu})$ graphically. The main idea is to use our proof for Theorem~\ref{thm:U0Convex}~(Appendix~\ref{app:StructuralResults})~to establish the following \underline{three-step method}:
 \\
 \noindent {\bf (Step 1)}~Compute $\onu^{\phi_\tau}$ and $\bar{\mathscr{U}}(\phi_\tau)$ for all $\tau$ in $\{1,\ldots,n_s+1\}$. 
 \\
 \noindent {\bf (Step 2)}~Identify the convex hull of the set $\big\{ \big ( \onu^{\phi_\tau}, \bar{\mathscr{U}}(\phi_\tau) \big) : \tau \in  \{1,\ldots,n_s+1\} \big\}$. 
 \\
 \noindent {\bf (Step 3)}~Determine $\bar{\mathscr{U}}^{0}_{\mathbb{L}}:[0,\onu^*] \rightarrow [0,1]$ as the lower boundary of the convex hull.

We will use the following example to illustrate our method, and to motivate the observations at the end of this section.
\begin{exmp}
\label{ex:graphicalmethod}
Consider that the system is characterized by $n_s=5$ and the following transition probabilities for $S_k$, which approximate the differential equation that describes the server state evolution in \cite{Savla2012A-Dynamical-Que}:
\beqan
\rho_{s,s+1} \Eqdef \frac{1}{5}\bigg(1-\frac{s-1}{n_s-1}\bigg), 
\ 
\rho_{s,s-1} \Eqdef \frac{1}{5}\bigg(\frac{s-1}{n_s-1}\bigg)
\eeqan
 The service rate function $\big( \mu(1),\ldots,\mu(5) \big )$ is set to be $ ( 0.01, 0.5, 0.2, 0.5, 0.05 )$. 
\end{exmp}

We now proceed to apply the graphical method to our example.
The following table lists the results obtained from step~1.

\begin{table}[h!]
\centering
\begin{tabular}{|c||c|c|c|c|c|c|}
\hline 
{\large $\tau$} & 1 & 2 & 3 & 4 & 5 & 6 \\ \hline
{\large $\onu^{\phi_{\tau}}$} 
    & 0.0000 & 0.0347 & 0.1993 
    & 0.1947 & 0.3000 & 0.0500  \\  \hline
\vspace{-0.07in} & & & & & & \\
{\large $\bar{\mathscr{U}}(\phi_{\tau})$}
    & 0.0000 & 0.2383 & 0.4309 
    & 0.6316 & 0.8571 & 1.0000 \\ \hline
\end{tabular}
\vspace{.1in}
\caption{Results of step~1 applied to Example~\ref{ex:graphicalmethod}}
\end{table}


The pairs $\big ( \onu^{\phi_\tau}, \bar{\mathscr{U}}(\phi_\tau) \big)$, for $\tau$ in $\{1,\ldots,n_s+1\}$, are the centers of the dark-red circles in the following figure, and the shaded area is their convex hull, whose construction is step~2. Finally, as described in step~3 and represented in the figure, the lower boundary of the convex hull is $\bar{\mathscr{U}}^{0}_{\mathbb{L}}:[0,\onu^*] \rightarrow [0,1]$.
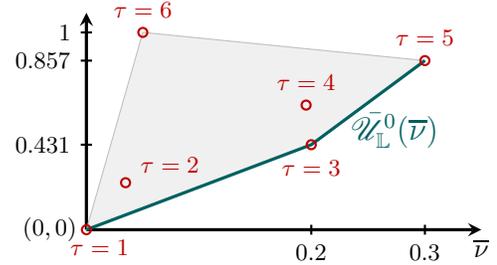
\begin{center}


\begin{tikzpicture}[xscale=15, yscale=2.625]

\coordinate (T1) at (0,0);
\coordinate (T2) at (0.0347,0.2383);
\coordinate (T3) at (0.1993,0.4309);
\coordinate (T4) at (0.1947,0.6316);
\coordinate (T5) at (0.300,0.8571);
\coordinate (T6) at (0.05,1);

\draw[->,very thick,>=stealth] (0,0) -- (0.35,0) node[anchor=north]{ $\onu$};

\draw[->,very thick,>=stealth] (0,0) -- (0,1.1);

\node[anchor=east ] at (0,0) {\small $(0,0)$};

\draw[thick] (0.005,1) -- (-0.005,1) node[anchor=east]{\small $1$};

\draw[thick] (0.005,0.8571) -- (-0.005,0.8571) node[anchor=east]{\small $0.857$};

\draw[thick] (0.005,0.4309) -- (-0.005,0.4309) node[anchor=east]{\small $0.431$};

\draw[thick] (0.300,0.025) -- (0.300,-.025) node[anchor=north]{\small $0.3$};

\draw[thick] (0.1993,0.025) -- (0.1993,-.025) node[anchor=north]{\small $0.2$};

\draw[fill=black!6,draw=black!25,thin] (T1) -- (T3) -- (T5) -- (T6) -- (T1);

\draw[very thick,teal!75!black] (T1) -- (T3) -- (T5) node[midway,below,xshift=10]{\large $\bar{\mathscr{U}}^{0}_{\mathbb{L}}(\overline{\nu})$};

\draw[thick,red!75!black] (T1) ellipse (0.004 and 0.022857) node[anchor=north,xshift=5] at (T1) {\small $\tau=1$};

\draw[thick,red!75!black] (T3) ellipse (0.004 and 0.022857) node[anchor=north,yshift=-2] at (T3) {\small $\tau=3$};

\draw[thick,red!75!black] (T5) ellipse (0.004 and 0.022857) node[anchor=south,yshift=2] at (T5) {\small $\tau=5$};

\draw[thick,red!75!black] (T2) ellipse (0.004 and 0.022857) node[anchor=south west,xshift=2] at (T2) {\small $\tau=2$};

\draw[thick,red!75!black] (T4) ellipse (0.004 and 0.022857) node[anchor=south,yshift=2] at (T4) {\small $\tau=4$};

\draw[thick,red!75!black] (T6) ellipse (0.004 and 0.022857) node[anchor=south,yshift=2] at (T6) {\small $\tau=6$};
\end{tikzpicture}

\end{center}



Our analysis for this example also leads to the following observations. (i)~As established by \cite[(8) and Theorems~1 and~2]{Lin2019Scheduling-task}, $\onu^{\phi}$ is maximized by a threshold policy. For our example, $\onu^*$ is $0.3$, which is achieved for $\phi_\tau$ when $\tau$ is $5$. (ii)~The corner points of $\bar{\mathscr{U}}^{0}_{\mathbb{L}}$ are among the pairs obtained in step~1. (iii)~As our example illustrates, $\onu^{\phi_{\tau}}$ is not necessarily monotonic with respect to $\tau$.

\section{Key Distributional Convergence Results: A Potential-like Approach}

\label{sec:DistributionalConvergence}

We start with the following lemma that is applicable for any time-homogeneous finite Markov chain. It establishes the existence of a potential-like function that will be useful later on. The proof of the lemma is deferred to Appendix~\ref{subsec:ProofOfLemmaPotential}.

\begin{lemma} \label{lem:potential} 
Let a time-homogeneous Markov chain $M:=\{M_k : 
k \in \mathbb{N} \}$ taking values in a finite 
set $\mathbb{M}$ and a reward function 
$\mathcal{R} : \mathbb{M} \times \mathbb{M} 
\rightarrow \mathbb{R}_+$ be given. If $M$ has 
a unique recurrent communicating class, there exists 
a map $\mathscr{H}:\mathbb{M} \rightarrow 
\mathbb{R}_+$, which we designate as 
potential-like, for which the following holds 
for every $m$ in $\mathbb{M}$:
\begin{eqnarray}
&& \myhb 
\EX \Big[ \mathcal{R}(M_{k+1},M_k) \ | \ M_k = m \Big] 
	\label{eq:PF1} \\ 
&& \myhb 
= \EX \Big[ \mathscr{H}(M_{k+1}) - \mathscr{H}(M_{k}) 
	\ | \ M_k = m \Big] + r_{avg},
	\nonumber
\end{eqnarray} 
where the average reward $r_{avg}$ can be computed 
using the stationary PMF $\varrho^M :  
\mathbb{M} \rightarrow [0,1]$ of $M$ as 
\beqan
r_{avg} 
:= \sum_{m \in \mathbb{M}} 
	\EX \Big[ \mathcal{R}(M_{k+1},M_k) 
		\ | \ M_k = m \ \Big] \varrho^M(m).
\eeqan
\end{lemma}


The following lemma is the first step towards proving Theorem~\ref{thm:MainDistConv}, which is the main result of this section.

\begin{lemma}	\label{lem:QGoesToZero} 
Let $\lambda$ in $(0,\onu^*)$ and $\epsilon$ in 
$(0,1)$ be given. If $\Phi^{\epsilon}_R(\lambda)$ 
is nonempty, there is a positive constant 
$\beta_{\lambda,\epsilon}$ such that the following 
inequality holds for every  $\onu \in (\lambda,
\onu^*)$:
\beqa
\label{eq:ZeroQueueProbBound}
\sum_{s \in \mathbb{S}} \pi^{\theta} (s,\cA,0) \leq \frac{(\onu - \lambda)} {\beta_{\lambda,\epsilon}}, \quad \ \theta \in \mathscr{X} \Phi_R^{\epsilon}(\onu)
\eeqa 
\end{lemma}

Before we proceed with the proof of 
Lemma~\ref{lem:QGoesToZero}, we note that one 
should expect it to be somewhat involved because 
it needs to ascertain that the inequality 
in~(\ref{eq:ZeroQueueProbBound}) holds (uniformly) 
for all policies in $\mathscr{X} \Phi_R^{\epsilon}(\onu)$. 
We decided to include the proof below, as opposed to 
deferring it to an appendix, because we find it to 
involve an instructive use of a potential-like 
function guaranteed by Lemma~\ref{lem:potential} 
to exist.

\begin{proof}
Select $\onu$ in $(\lambda,\onu^*)$, and let $\phi$ be any policy in $\Phi_R^{\epsilon}(\onu)$, which we know from Proposition~\ref{prop:PhiEpsilonNonEmpty} is nonempty, and set $\theta = \mathscr{X}(\phi)$.
Recall that $\bX^{\theta}$ is stable by Proposition~\ref{prop:SufficientCondStab}.  
In our proof we will make use of Lemma~\ref{lem:potential} by selecting $M = \obY^{\phi}$ and $\mathcal{R}(\by',\by) = \mu(s)$ for all $\by'$ and $\by$ in $\mathbb{Y}$, where we recall that $\by:=(s,w)$. We define $s^* = \arg \max_{s \in \mathbb{S}} \mathscr{H}(s,\cA)$, where $\mathscr{H}$ is the potential-like map obtained from Lemma~\ref{lem:potential} for the aforementioned choices of $M$ and $\mathcal{R}$.

The following hitting time will be central in our proof:
\begin{equation*}
T_{\bx}^{\theta} 
:= \min \{k \geq 1 \ | \ \bX_k^{\theta}=(s^*,\cA,0), 
	\ \bX_0=\bx \},
\end{equation*} 
where we adopt the convention that $T^\theta_{\bx}$ is infinite if $\bX^\theta_k=(s^*,\cA,0)$ never occurs for $k \geq 1$. We will also use the following lower bound:
\begin{equation*}
\underbar{T}_{\bx}^{\theta} 
:= \min \{k \geq 1 \ | \ \mathscr{V}(\bX_k^{\theta}) 
	\leq v^*, \ \bX_0=\bx \},
\end{equation*} 
where $\mathscr{V}(\bx) : = q + \mathscr{H}(\by)$ and $v^* : = \mathscr{H}(s^*,\cA)$. Here, we also adopt the convention that $\underbar{T}_{\bx}^{\theta}$ is infinite if $\mathscr{V}(\bX_k^{\theta}) \leq v^*$ never occurs for $k \geq 1$. Notice that since $\mathscr{V}(s^*,\cA,0) = v^*$, we have
$\underbar{T}_{\bx}^{\theta} \leq T_{\bx}^{\theta}, 
	\quad \bx \in \mathbb{X}$.

Subsequently, we use $T_{\bx}^{\theta}$, 
$\underbar{T}_{\bx}^{\theta}$ and $\mathscr{V}$ 
to obtain a lower bound for 
$\EX[T^\theta_{(s^*,\cA,0)}]$~-~the return time of 
$(s^*,\cA,0)$~-~ which will ultimately lead to 
the proof of~(\ref{eq:ZeroQueueProbBound}).

As we argue subsequently, the following lower bound 
for $\EX[\underbar{T}^\theta_{(s^*,\cA,1)}]$, which 
we will derive later in this proof, leads to~(\ref{eq:ZeroQueueProbBound}) almost immediately:
\begin{equation}
\label{eq:lowerbdEXTy1}
\EX \big[ \underbar{T}_{(s^*,\cA,1)}^{\theta} \big]
\geq \frac{1}{\onu-\lambda}
\end{equation}

We start by using the law of total probability to conclude 
that the following inequality holds:
\begin{eqnarray*}
&& \hspace{-0.25in} 
\EX \big[ T_{(s^*,\cA,0)}^{\theta} \big]  
	\label{eq:lowerbdEXTx-1} \\ 
&& \hspace{-0.25in}
\geq (1 + \EX[T_{(s^*,\cA,1)}^{\theta}] ) 
	P_{\bx_1^{\theta} | \bX^\theta_0} \big( (s^*,\cA,1) | (s^*,\cA,0) \big) 
	\nonumber
\end{eqnarray*} 
which after substituting~(\ref{eq:lowerbdEXTy1}) and using 
the fact that $P_{\bX^\theta_1 | \bX^\theta_0} \big( (s^*,\cA,1) \ | \ (s^*,\cA,0) \big) = \lambda (1-\rho_{s^*,s^*-1})$ leads to:
\begin{equation}
\label{eq:lowerbdEXTx-2}
\EX \big[ T_{(s^*,\cA,0)}^{\theta} \big] 
\geq  (1-\rho_{s^*,s^*-1}) 
	\frac{1+\onu-\lambda}{\onu/\lambda-1}
\end{equation}

According to \cite[(3)~Theorem, p. 227]{Grimmett2001Probability-and}, (\ref{eq:lowerbdEXTx-2})  implies that:
\begin{equation}
\label{eq:upperbdPix-1}
\pi^{\theta}(s^*,\cA,0) 
\leq \frac{\onu/\lambda-1}{1-\rho_{s^*,s^*-1}}
\end{equation}
At this point we intend to use the following inequality to relate 
$\pi^{\theta}_{\lambda} (s^*,\cA,0)$ with $\sum_{s \in \mathbb{S}} 
\pi^{\theta}_{\lambda} (s,\cA,0)$ :
\begin{eqnarray*}
&& \hspace{-0.25in} \pi^{\theta} (s^*,\cA,0) 
	\label{eq:lowerbdPix-2} \\ 
&& \hspace{-0.25in} \geq \sum_{s \in \mathbb{S}} \pi^{\theta}(s,\cA,0) 
	P \Big ( \bX_{k+2n_s}^{\theta} = (s^*,\cA,0) 
		 \big|  \bX_{k}^{\theta} = (s,\cA,0) \Big ) 
	\nonumber
\end{eqnarray*}

Recall from Proposition \ref{prop:SufficientCondStab} 
that $\bX^{\theta}$ is irreducible.
Moreover, we can show that there is positive 
$\tilde{\beta}_{\lambda, \epsilon}$ satisfying
\begin{equation}
\label{eq:lowerbdCondP}
P \Big ( \bX_{k+2n_s}^{\theta} 
= (s',\cA,0) \ \big | \ \bX_{k}^{\theta} 
= (s,\cA,0) \Big ) \geq \tilde{\beta}_{\lambda,\epsilon}
\end{equation} 
for all $\theta$ in 
$\mathscr{X} \Phi_R^{\epsilon}(\onu)$.
For example, one can verify
\begin{eqnarray}
&& \hspace{-0.25in}
\tilde{\beta}_{\lambda,\epsilon}
= \epsilon \lambda (1-\lambda)^{2n_s}
\min_{i \in \mathbb{S}} \big( 1-\mu(i) \big)^{2{n_s}}
	\min_{j \in \mathbb{S}} \mu(j)
	\label{eq:BetaTildeDef} \\
&& \hspace{-0.25in}
\times \big( \Pi_{i=1}^{n_s-1} \rho_{i+1,i}\rho_{i,i+1} \big)\min_{i\in\mathbb{S}} \big((1-\rho_{i,i+1})(1- \rho_{i,i-1}) \big)^{2{n_s}}
	\nonumber
\end{eqnarray}
satisfies the inequality in \eqref{eq:lowerbdCondP}. 
The proof of (\ref{eq:ZeroQueueProbBound}) follows 
from~(\ref{eq:upperbdPix-1}) and (\ref{eq:lowerbdCondP}) 
with $\beta_{\lambda,\epsilon} : = \lambda (1 -
\rho_{s^*,s^*-1})\tilde{\beta}_{\lambda,\epsilon}$.

\paragraph*{\underline{Proof of~(\ref{eq:lowerbdEXTy1})}} We now proceed to proving that~(\ref{eq:lowerbdEXTy1}) holds. We start with the following equalities that hold for any $\bx$ satisfying $\mathscr{V}(\bx) > v^*$, 
which implies $q > 0$:
\begin{eqnarray}
\label{eq:Theorem3-eq-1}
&& \myhb 
\EX [ Q_k^{\theta} - q | \bX_{k-1}^{\theta} = \bx] = \lambda - \phi(s,w) \mu(s) \\
&& \myhb \EX [ \mathscr{H}(\bY_k^{\theta}) - \mathscr{H}(\by) | \bX_{k-1}^{\theta} = \bx] 
	\label{eq:Theorem3-eq-2} \\ 
&& = \EX [ \mathscr{H}(\obY_k^{\phi}) 
	- \mathscr{H}(\by) | \obY_{k-1}^{\phi} = \by] 
	\nonumber \\ 
&& \overset{\mbox{(i)}}{=} \phi(s,w) \mu(s) - \onu 
	\nonumber
\end{eqnarray}
In proving (\ref{eq:Theorem3-eq-1}) and (\ref{eq:Theorem3-eq-2}),  we used the fact that if $\mathscr{V}(\bx) > v^*$ holds, $q \geq 1$, which, since $\theta = \mathscr{X}(\phi)$, implies that the policy $\phi$ is applied. In addition, we used Lemma~\ref{lem:potential} to establish (i), where we used the fact that, for our choices of $M$ and $\mathcal{R}$, $r_{avg}$ is $\onu$. Hence, by adding the terms of (\ref{eq:Theorem3-eq-1}) and (\ref{eq:Theorem3-eq-2}) we conclude that the following holds when $\bx$ is such that $\mathscr{V}(\bx) > v^*$ holds:
\begin{equation}
\label{eq:Theorem3-eq-3}
\EX [ \mathscr{V}(\bX_k^{\theta}) - \mathscr{V}(\bx) | \bX_{k-1}^{\theta} = \bx] = \lambda - \onu .
\end{equation}

Because $\underbar{T}_{(s^*,\cA,1)}^{\theta} \geq k$ implies $\mathscr{V}(\bX_{k-1}^{\theta}) > v^*$, from 
\eqref{eq:Theorem3-eq-3}
we obtain
\begin{eqnarray}
&& \myhb \sum_{k=1}^{\infty} \EX \Big [ \big ( \mathscr{V}(\bX_k^{\theta}) - \mathscr{V}(\bX_{k-1}^{\theta}) \big ) \mathcal{I}_{\underbar{T}_{(s^*,\cA,1)}^{\theta} \geq k} \Big | \bX_0=(s^*,\cA,1) \Big ] \nonumber \\
\myeq \sum_{k=1}^{\infty} (\lambda - \onu)
	P \big( \underbar{T}_{(s^*,\cA,1)}^{\theta} \geq k 
		\big| \bX_0 = (s^*, \cA, 1) \big)
	\label{eq:Theorem3-eq-6} \\
\myeq (\lambda - \onu) 
	\EX [\underbar{T}_{(s^*,\cA,1)}^{\theta}  ], 
	\nonumber
\end{eqnarray} 
where $\mathcal{I}_{\underbar{T}_{(s^*,\cA,1)}^{\theta} \geq k}$ is $1$ when $\underbar{T}_{(s^*,\cA,1)}^{\theta} \geq k$ holds, and is $0$ otherwise.

From the definition of $\mathscr{V}(\bx) = q 
+ \mathscr{H}(\by)$, conditional on $\{ \bX_0 = 
(s^*, \cA, 1) \}$, for every $K$ in $\N$, 
\beqan
&& \myhb \sum_{k=1}^K \Big( \mathscr{V}(\bX_k^{\theta}) - \mathscr{V}(\bX_{k-1}^{\theta}) \Big ) \mathcal{I}_{\underbar{T}_{(s^*,\cA,1)}^{\theta} \geq k}  \lb 
\myleq \min\big\{ K, 
	\underbar{T}_{(s^*,\cA,1)}^{\theta}  \big\} 
	+ \max_{\by \in \mathbb{Y}} \mathscr{H}(\by)
	- v^* \lb 
\myleq \underbar{T}_{(s^*,\cA,1)}^{\theta} 
	+ \max_{\by \in \mathbb{Y}} \mathscr{H}(\by)
	- v^*. 
\eeqan
Since $\bX^\theta$ is positive recurrent, 
$\EX\big[ \underbar{T}_{(s^*,\cA,1)}^{\theta} \big]
< \infty$. Therefore, the dominated convergence
theorem \cite[pp.179-180]{Grimmett2001Probability-and} allows
us to interchange the order of the summation and 
the expectation in \eqref{eq:Theorem3-eq-6}. 
After the interchange, we have
\beqan
&& \myhb \EX \Big[ \sum_{k=1}^\infty \big ( \mathscr{V}(\bX_k^{\theta}) - \mathscr{V}(\bX_{k-1}^{\theta}) \big ) \mathcal{I}_{\underbar{T}_{(s^*,\cA,1)}^{\theta} \geq k} \Big | \bX_0=(s^*,\cA,1) \Big ] \lb
\myeq \EX \Big[ \mathscr{V}\Big( \bX_{\underbar{T}_{(s^*,\cA,1)}^{\theta}}^{\theta} \Big) 
- \mathscr{V}( \bX_0 ) \ \Big| \ \bX_0=(s^*,\cA,1) \Big ] \lb
\myeq (\lambda - \onu) \EX [\underbar{T}_{(s^*,\cA,1)}^{\theta}  ],
\eeqan
where $\bX_{\underbar{T}_{(s^*,\cA,1)}^{\theta} } ^{\theta}$ is $\bX_k^{\theta}$ at time $k = \underbar{T}_{(s^*,\cA,1)}^{\theta}$. 

The equality above leads to~(\ref{eq:lowerbdEXTy1}) once we realize that the following inequality holds:
\begin{eqnarray*}
&& \myhb \EX \Big [ \mathscr{V} \Big ( \bX_{\underbar{T}_{(s^*,\cA,1)}^{\theta} } ^{\theta} \Big ) 
	- \mathscr{V}(\bX_0) \ \Big | \ 
		\bX_0 = (s^*,\cA,1) \Big ]  \\ 
&& \myhb = \EX \Big [ \mathscr{V} \Big ( 
	\bX_{\underbar{T}_{(s^*,\cA,1)}^{\theta} }^{\theta} 
		\Big )  \ \Big | \ 
			\bX_0 = (s^*,\cA,1) \Big ] - (v^*+1) \\ 
&& \myhb \leq v^* - (v^*+1) = -1
\end{eqnarray*}
\end{proof}

\begin{theorem}  
\label{thm:MainDistConv} 
Let $\lambda$ in $(0,\onu^*)$ and 
$\epsilon$ in $(0,1)$ be given. Suppose 
$\Phi^{\epsilon}_R(\lambda)$ 
is nonempty, and let  
$\beta_{\lambda, \epsilon}$ be a positive constant
satisfying Lemma~\ref{lem:QGoesToZero}. Then, 
there is a positive constant $\eta_{\epsilon}$
such that the following inequality holds for all 
$\onu \in (\lambda, \onu^*)$:
for every $\phi$ in $\Phi_R^{\epsilon}(\onu)$, 
\begin{eqnarray}
&& \myhb \sum_{\by \in \mathbb{Y} } 
	\Bigg | \overline{\pi}^{\phi}(y) 
	- \sum_{q > 0 } \pi^{\mathscr{X}(\phi)}(y,q) 
	\Bigg |  
	\label{eq:MainDistConvIneq} \\ 
&& \myhb \leq \frac{\beta_{\lambda,\epsilon} 
	+ \eta_\epsilon}
	{\beta_{\lambda,\epsilon}}
		(\onu-\lambda)^{\frac{1}{2}}
	+ \frac{3}{\beta_{\lambda,\epsilon}}(\onu-\lambda).
	\nonumber
\end{eqnarray}
\end{theorem}

In Appendix~\ref{app:ProofOfMainDistConv}, we provide 
a proof for Theorem~\ref{thm:MainDistConv} where the existence of $\beta_{\lambda,\epsilon}$ follows from Lemma~\ref{lem:QGoesToZero}. As was 
the case with Lemma~\ref{lem:QGoesToZero}, but even 
more so, the proof of Theorem~\ref{thm:MainDistConv} 
is rather involved because the inequality 
in~(\ref{eq:MainDistConvIneq}) must hold uniformly 
for all $\phi$ in $\Phi_R^{\epsilon}(\onu)$. The 
following corollary is an immediate consequence of 
Theorem~\ref{thm:MainDistConv}.

\begin{coro}
\label{coro:MainUDistConverence} Let $\lambda$ in 
$(0,\onu^*)$ and $\epsilon$ in $(0,1)$ be given.  
Suppose $\Phi^{\epsilon}_R(\lambda)$ 
is nonempty, and let  
$\beta_{\lambda, \epsilon}$ be a positive constant
satisfying Lemma~\ref{lem:QGoesToZero}. 
Then, there 
is a positive constant $\eta_{\epsilon}$ such 
that the following inequality holds for all $\onu 
\in (\lambda,\onu^*)$: for every 
$\phi \in \Phi_R^{\epsilon}(\onu)$, 
\begin{eqnarray*}
&& \myhb \Big | \bar{\mathscr{U}}(\phi) 
	- \mathscr{U}(\lambda,\mathscr{X}(\phi))  \Big | 
	\\
&& \myhb \leq \frac{\beta_{\lambda,\epsilon} 
 	+ \eta_\epsilon}{\beta_{\lambda,\epsilon}}
 		(\onu-\lambda)^{\frac{1}{2}}
 + \frac{3}{\beta_{\lambda,\epsilon}}
 	(\onu-\lambda).
 	\nonumber
\end{eqnarray*}
\end{coro}
We remind the reader that $\bar{\mathscr{U}}(\phi)$
implicitly depends on $\lambda$ via the stationary
PMF $\overline{\pi}^\phi$
in Definition~\ref{defn:UtilizationRate}.

\section{Conclusions}
\label{sec:conclusions}

We put forth a methodology to design policies that schedule tasks from a queue to a server whose performance depends on the scheduling history. Our approach builds on and generalizes previous work that sought to design stabilizing non-preemptive policies. This article introduces methods to design non-preemptive policies that are not only stabilizing but also lessen the so-called utilization rate, which accounts for the proportion of time the server is working. Given a rate of arrival of tasks at the queue, our two main results yield a tractable method to compute the infimum of the utilization rates that are attainable by all stabilizable non-preemptive policies, and characterize subsets of conveniently-structured policies whose utilization rate is arbitrarily close to the infimum.



\appendix
\subsection{Structural Results for $\bar{\mathscr{U}}^{0}_{\mathbb{L}}(\overline{\nu})$ and Proof of Theorem~\ref{thm:U0Convex}}
\label{app:StructuralResults}
Define $\Phi^\dagger$ to be set of 
policies in $\Phi_R$ which are deterministic
except for at most at one state where the
policy randomizes between two admissible
actions. In other words,
\beqan
\Phi^{\dagger} 
\Eqdef \Big\{ \phi \in \Phi_R 
	& \myb \Big| & \myb \mbox{there is } \mathbb{S}_\phi
	    \subset \mathbb{S} \mbox{ such that (i) } 
	     |\mathbb{S} \setminus \mathbb{S}_\phi | 
		\leq 1 \\ 
	&& \hspace{-0.4in} 
	\mbox{ and (ii) }
	\phi(\os,\ow)\in\{0,1\} \mbox{ for all } \os \in 
	    \mathbb{S}_\phi \Big\}.
\eeqan


For each $\phi$ in $\Phi_R$, let
$\overline{\boldsymbol{\Pi}}(\phi)$ be the set 
of stationary PMFs of $\overline{\bY}^\phi$. The proof of~\cite[Theorem 4.4]{Altman1999Constrained-Mar}
tells us that, given a non-empty solution set $\mathbb{L}^{0}(\overline{\nu})$ for LP~\eqref{LP-Definition}, there exist an optimal occupation measure $\ell^* \in \mathbb{L}^{0}(\overline{\nu})$, a policy $\phi^*\in\Phi^\dagger$, and a stationary PMF $\opi^*\in\overline{\boldsymbol{\Pi}}(\phi^*)$ such that the following equalities hold:
\beqan
\ell^*_{\oby,\mathcal{R}}+\ell^*_{\oby,\mathcal{W}} \myeq \opi^*(\oby)\\
\ell^*_{\oby,\mathcal{W}} \myeq \opi^*(\oby)\phi^*(\oby), \  \oby \in \overline{\mathbb{Y}}
\eeqan
Hence, we can rewrite  ${\bar{\mathscr{U}}^0_{\mathbb{L}}}$ as

\begin{equation}
\label{eq:ReduceMinimize}
\begin{matrix}
\displaystyle \bar{\mathscr{U}}^{0}_{\mathbb{L}}(\overline{\nu}) =&\underset{\overline{\bpi} \in 
	\overline{\boldsymbol{\Pi}}(\phi), \phi\in\Phi^{\dagger}}{\min} & \sum_{\overline{\by} \in \mathbb{Y}} 
	\overline{\pi}(\overline{\by}) 
	\phi(\overline{\by})  \\
&\textrm{s.t.} & \sum_{\overline{\by} \in \mathbb{Y}} 
    \mu(\overline{s}) 
	\overline{\pi}(\overline{\by}) 
	\phi(\overline{\by}) 
	= \onu
\end{matrix}
\end{equation}


We shall further divide $\Phi^{\dagger}$ into three subsets where the probabilities to choose to work at the state $(1,\cA)$ are one, between zero and one, or zero and consider the LP \eqref{eq:ReduceMinimize} on each of the subsets in Lemmas~\ref{lem:linear1} through \ref{lem:linear3}. Before we proceed with the proof, we restate the definition of threshold policies from \cite{Lin2019Scheduling-task}.

We define a threshold policy $\phi_{\tau}$ as
\beqan
\phi_{\tau}(s,w)
\Eqdef \begin{cases}
        0 &\mbox{if $s\geq\tau$ and $w=\cA$,}\lb
        1 &\mbox{otherwise.}
      \end{cases}
\eeqan
 
\begin{lemma}
\label{lem:linear1}
For every $\phi \in \Phi^{\dagger}$ with 
$\phi(1,\cA) = 1$, there exist $\tau_1, \tau_2 
\in \mathbb{S}\cup\{n_s+1\}$ and $\alpha \in
[0,1]$ such that
\beqan
\onu^\phi 
\myeq (1-\alpha)\onu^{\phi_{\tau_1}}
	+ \alpha\onu^{\phi_{\tau_2}}, \lb
\bar{\mathscr{U}}(\phi) 
\myeq (1-\alpha)\bar{\mathscr{U}}(\phi_{\tau_1}) 
	+ \alpha\bar{\mathscr{U}}(\phi_{\tau_2}).
\eeqan
\end{lemma}
\begin{proof}
We define the mapping 
$\mathcal{T}: \Phi_R
\to \mathbb{S} \cup \{0\}$, where 
\beqan
\mathcal{T}(\phi)
\Eqdef \max\{ \os \in \mathbb{S} \ | \ 
	\phi(\os, \cA) = 1 \}, \  \phi \in \Phi_R. 
\eeqan
We assume that $\mathcal{T}(\phi) = 0$ if the
set on the RHS is empty. We first observe that 
$\mathcal{T}(\phi)\geq 1$ since $\phi(1,\cA) 
= 1$ and the only positive recurrent 
communicating class is $\{\oby\in\mathbb{Y} \ | \
\os \geq \mathcal{T}(\phi)\}$. Second, consider
the following policy $\phi'$: 
\beqan
\phi'(\oby) = 
    \begin{cases}
        \phi(\oby) &\mbox{if $\os\geq\mathcal{T}(\phi)$}\lb
        1          &\mbox{otherwise}
    \end{cases}
\eeqan
It is clear that $\phi'$ has the same 
service rate and utilization rate as $\phi$
because both policies have the same positive 
recurrent communicating class and the policies 
inside the class are identical. 

Recall that there is only one state, 
say $s'$, where $\phi$ randomizes 
between two actions. Thus, if 
$s' < \mathcal{T}(\phi)$, $\phi'$ is just a 
threshold policy $\phi_{\mathcal{T}(\phi)+1}$.
On the other hand, if $s' > \mathcal{T}(\phi)$, 
$\phi'$ is of the following form:
\begin{equation}
\phi'(\oby) = 
    \begin{cases}
        \gamma &\mbox{if $\ow=\cA$ and $\os=s'$}\label{eq:form1}\\
        \phi_{\cT(\phi)+1}(\oby) &\mbox{otherwise}
    \end{cases}
\end{equation}
Suppose that $\tau_1 = \cT(\phi)+1$ and $\tau_2 
= s'+1$. 
We rewrite $\gamma$ in \eqref{eq:form1} as 
\beqa
\myb \gamma
\myeq \frac{ \alpha \cdot \opi^{\phi_{\tau_2}}
		(\tau_2-1,\cA)}
	{\alpha \cdot \opi^{\phi_{\tau_2}}
		(\tau_2-1,\cA)
	+(1-\alpha) \opi^{\phi_{\tau_1}}
		(\tau_2-1,\cA)}
	\label{eq:gamma}
\eeqa
for some $\alpha \in (0, 1)$. Note that, for 
every $\gamma \in (0,1)$, we can find an 
appropriate $\alpha\in(0,1)$ that satisfies
\eqref{eq:gamma}
because $\opi^{\phi_{\tau_1}}(\tau_2 - 1,
\cA)>0$ and $\opi^{\phi_{\tau_2}}(\tau_2 - 1,
\cA)>0$ from the fact that
$\cT(\phi) < s'$. 

By solving the global balance equations
for $\obY$ under the policy 
$\phi'$, we get the following stationary PMF.
Its derivation is provided in  
Appendix~\ref{appn:disSplit}: for every 
$\oby$ in $\mathbb{Y}$, 
\beqa
\label{eq:distributionSplit}
\opi^{\phi'}(\oby) 
\myeq (1-\alpha) \opi^{\phi_{\tau_1}}(\oby)
+ \alpha \cdot \opi^{\phi_{\tau_2}}(\oby)
\eeqa

The service rate can 
be obtained using the stationary PMF. 
\beqan
\overline{\nu}^{\phi'}
\myeq
	\sum_{\overline{\by} \in \mathbb{Y}} 
	\mu(\overline{s}) 
	\ \opi^{\phi'}
		(\overline{\by}) 
	\ \phi'
		(\overline{\by})
\eeqan
Substituting the RHS of \eqref{eq:distributionSplit}
for $\opi^{\phi'}(\oby)$, 
we obtain
\beqa
\overline{\nu}^{\phi'}
\myeq \sum_{\overline{\by} \in \mathbb{Y}} 
	\Big( \mu(\overline{s}) 
	\big( \alpha \cdot 
		\opi^{\phi_{\tau_2}}(\oby) 
	+ (1-\alpha)\opi^{\phi_{\tau_1}}(\oby) \big) 
		\phi'
			(\overline{\by}) \Big) \lb
\myeq \mu(\tau_2-1) 
	\big( \alpha \cdot \opi^{\phi_{\tau_2}}
		(\tau_2-1,\cA) \lb
	&& \myhf + (1-\alpha) \opi^{\phi_{\tau_1}}
		(\tau_2-1,\cA) \big)
	\phi'
		(\tau_2-1, \cA) \lb
	&& \myb + \sum_{\overline{\by} \in \mathbb{Y}
		\setminus \{(\tau_2-1,\cA)\}} 
	\Big( \mu(\overline{s}) \big( 
		\alpha \cdot \opi^{\phi_{\tau_2}}(\oby) 
		\label{eq:lemma7-1} \\
	&& \myhf + (1-\alpha) 
		\opi^{\phi_{\tau_1}}(\oby) \big)
	\phi'(\oby)
		\Big).
	\nonumber
\eeqa
Using the definition of $\phi'$ in \eqref{eq:form1}
and the fact that $\mathcal{T}(\phi) < s'$, 
\beqa
\myhb \eqref{eq:lemma7-1}
\myeq \mu(\tau_2 - 1) 
	\big( \alpha \cdot \opi^{\phi_{\tau_2}}
		(\tau_2-1,\cA) \lb
	&& \hspace{0.15in} 
		+ (1-\alpha) \opi^{\phi_{\tau_1}}
		(\tau_2-1,\cA) \big) \lb
	&& \myb \times 
	\big((1-\gamma)\phi_{\tau_1}(\tau_2-1,\cA) 
		+ \gamma\phi_{\tau_2}(\tau_2-1,\cA) \big)
	\label{eq:lemma7-2} \\
	&& \hspace{-0.2in}
	+ \sum_{\overline{\by} \in \mathbb{Y}
		\setminus \{(\tau_2-1,\cA)\} }
		\Big( \mu(\overline{s}) 
	\big( \alpha \cdot \opi^{\phi_{\tau_2}}(\oby)
		\lb 
	&& \hspace{0.8in} 
		+ (1-\alpha) \opi^{\phi_{\tau_1}}(\oby) 
			\big) \phi_{\tau_1}(\oby) \Big).
	\label{eq:lemma7-3}
\eeqa

First, using the expression in \eqref{eq:gamma} for 
$\gamma$ in the first term, we get
\beqan
\eqref{eq:lemma7-2}
\myeq \mu(\tau_2 - 1)\Big(
	(1-\alpha)\opi^{\phi_{\tau_1}}(\tau_2-1,\cA)
	\phi_{\tau_1}(\tau_2-1,\cA)\lb
	&& \hspace{0.15in} 
		+ \alpha\opi^{\phi_{\tau_2}}(\tau_2-1,\cA)
	\phi_{\tau_2}(\tau_2-1,\cA) \Big). 
\eeqan
Second, we conclude $\opi^{\phi_{\tau_2}}(\oby)
\phi_{\tau_1}(\oby) = \opi^{\phi_{\tau_2}}(\oby)
\phi_{\tau_2}(\oby)$ for all 
$\oby \in \mathbb{Y}
\setminus \{(\tau_2-1,\cA)\}$  
by considering the following three cases: (i)~If $\os \geq \tau_2$ and $w = \cA$, we have 
	$\phi_{\tau_1}(\os,\ow) = \phi_{\tau_2}(\os,\ow) = 0$ from the
	definition of $\phi_{\tau_1}$ and $\phi_{\tau_2}$. (ii)~If $\os < \tau_2 - 1$, then 
	 $\opi^{\phi_{\tau_2}}(\os,\ow) = 0$
	 (because $(\os, \ow)$ is transient). (iii)~If $\ow = \cB$, then 
	 $\phi_{\tau_1}(\os,\ow) = \phi_{\tau_2}(\os,\ow) = 1$.
As a result, 
\beqan
\eqref{eq:lemma7-3}
\myeq \sum_{\overline{\by} \in \mathbb{Y}
	\setminus \{(\tau_2-1,\cA)\} } 
	 \mu(\overline{s}) \Big(
		(1-\alpha)\opi^{\phi_{\tau_1}}(\oby)
		\phi_{\tau_1}(\oby) \lb
		&& \hspace{1.05in} 
		+ \alpha\opi^{\phi_{\tau_2}}(\oby)
		\phi_{\tau_2}(\oby)\Big). 
\eeqan
Summing \eqref{eq:lemma7-2} and 
\eqref{eq:lemma7-3}, we get
\beqa
\overline{\nu}^{\phi}
= \overline{\nu}^{\phi'}
\myeq \sum_{\overline{\by} \in \mathbb{Y}} 
	 \mu(\overline{s}) \Big(
		(1-\alpha)\opi^{\phi_{\tau_1}}(\oby)
		\phi_{\tau_1}(\oby) \lb
		&& \hspace{0.5in} 
		+ \alpha\opi^{\phi_{\tau_2}}(\oby)
		\phi_{\tau_2}(\oby)\Big)\lb
\myeq (1-\alpha)\onu^{\phi_{\tau_1}} + \alpha\onu^{\phi_{\tau_2}}. 
\eeqa
Following similar steps, we can show 
$\bar{\mathscr{U}}(\phi) = (1-\alpha)\bar{\mathscr{U}}(\phi_{\tau_1})$ $+ \alpha\bar{\mathscr{U}}(\phi_{\tau_2})$. Finally, we include $\alpha$ at zero and one for the Lemma statement to consider the case where $\phi$ is a deterministic policy without the randomization.
\end{proof}
\begin{lemma}
\label{lem:linear2}
For every $\phi\in\Phi^{\dagger}$ with $\phi(1,\cA) \in (0,1)$, there exist $\tau_2\in \mathbb{S}\cup\{n_s+1\}$ and $\beta\in[0,1]$ such that
\beqan
\onu^\phi 
\myeq \beta\onu^{\phi_{\tau_2}}
\ \mbox{ and } \ 
\bar{\mathscr{U}}(\phi) 
= \beta\bar{\mathscr{U}}(\phi_{\tau_2}).
\eeqan
\end{lemma}
\begin{proof}
Because $\phi$ randomizes between two
actions only at state $(1,\cA)$, $\phi$ 
is deterministic at all other states. There are
two cases to consider: (i) $\cT(\phi) > 0$ 
and $\cT(\phi) = 0$. In the first case, 
$\phi$ has the same 
service rate and utilization rate as 
the threshold 
policy $\phi_{\cT(\phi)+1}$.
In the second case,
\beqan
\phi(\oby) =\begin{cases}
        \gamma &\mbox{if $\oby = (1,\cA)$,}\lb
        1      &\mbox{if $\ow = \cB$,}\lb
        0      &\mbox{otherwise.}
      \end{cases}
\eeqan
The rest of the proof is identical to 
that of Lemma~\ref{lem:linear1} after 
replacing (a) $\phi_{\tau_2}$ with $\phi_{2}$ 
and (b) $\phi_{\tau_1}$ with 
$\phi_{1}$, which is a policy that always rests
with $\onu^{\phi_1} = 
\bar{\mathscr{U}}(\phi_{1}) = 0$.
\end{proof}

Before we state the final lemma, note that, 
when $\phi(1,\cA) = 0$, the process $\obY^\phi$ 
could have two positive recurrent communicating 
classes. For such a policy $\phi$, the 
utilization rate $\bar{\mathscr{U}}$ is not well 
defined. Hence, for a policy $\phi$ with 
$\phi(1, \cA) = 0$, we define a set 
of pairs consisting of a service rate and 
a utilization rate.
\beqan
&&\myhb\bar{{\bf SU}}(\phi) \lb
&&\myhb\Eqdef \left\{\Big(\sum_{\overline{\by} \in \mathbb{Y}} 
    \mu(\overline{s}) 
	\overline{\pi}(\overline{\by}) 
	\phi(\overline{\by}),\ \sum_{\overline{\by} \in \mathbb{Y}} 
	\overline{\pi}(\overline{\by}) 
	\phi(\overline{\by})\Big):\ \overline{\bpi} \in 
	\overline{\boldsymbol{\Pi}}(\phi)\right\}
\eeqan

\begin{lemma}
\label{lem:linear3}
For every $\phi\in\Phi^{\dagger}$ with $\phi(1,\cA) = 0$, there exist $\tau_1, \tau_2 \in \mathbb{S}\cup\{n_s+1\}$ and $\alpha\in[0,1]$ such that
\begin{eqnarray*}
\bar{{\bf SU}}(\phi)
\myeq \Bigg\{\beta\Big((1-\alpha)\onu^{\phi_{\tau_1}} 
    	+ \alpha\onu^{\phi_{\tau_2}}, \\
&& (1-\alpha) \bar{\mathscr{U}}(\phi_{\tau_1}) 
   	+ \alpha\bar{\mathscr{U}}(\phi_{\tau_2})\Big)
   	:\ \beta \in [0,1]\Bigg\}.
\end{eqnarray*}
\end{lemma}
\begin{proof}
If $\cT(\phi) = 0$ which implies that the policy always rest, it is clear that $(1,\cA)$ is an absorbing state and the service rate and the utilization rate are both zero. If $\cT(\phi) > 0$, we can represent $\phi$ as
\beqan
\phi(\oby) =\begin{cases}
        0 &\mbox{if $\oby = (1,\cA)$,}\lb
        \phi'(\oby) &\mbox{otherwise,}
      \end{cases}
\eeqan
where $\phi' = \phi_{\cT(\phi)+1}$ as in \eqref{eq:form1}. The MC now has two positive recurrent communicating classes, and the stationary PMF can be any convex combination of stationary PMFs of $\phi'$ and $\phi_1$. This is also true for utilization rate and service rate. 
\end{proof}
{\it Proof of Theorem~\ref{thm:U0Convex}: \quad} 
By Lemmas~\ref{lem:linear1} through \ref{lem:linear3}, the 
utilization rate and the service rate pair for every 
policy in $\Phi^\dagger$ can be written as a convex combination 
of rate pairs of at most two threshold policies and $(0,0)$. 
Hence, the optimization problem \eqref{eq:ReduceMinimize} can be 
transformed into a following optimization problem over two variables 
$\alpha,\beta\in[0,1]$ for convex combination and two thresholds
$\tau_1,\tau_2\in\mathbb{S}\cup\{n_s+1\}$:

\begin{equation*}
\begin{matrix}
\displaystyle \min_{\begin{matrix}\alpha,\beta\in[0,1]\\\tau_1,\tau_2\in\mathbb{S}\cup\{n_s+1\}\end{matrix}} & \beta\Big((1-\alpha)
	\bar{\mathscr{U}}(\phi_{\tau_1}) + \alpha\bar{\mathscr{U}}(\phi_{\tau_2})\Big)  \\
\textrm{s.t.} & \beta\Big((1-\alpha)\onu^{\phi_{\tau_1}} + \alpha\onu^{\phi_{\tau_2}}\Big)
	= \onu
\end{matrix}
\end{equation*}

It is clear from this argument that 
$\{ \bar{\mathscr{U}}^{0}_{\mathbb{L}}(\overline{\nu}): \overline{\nu} \in [0, \overline{\nu}^*] \}$ forms
the lower boundary of the {\it convex hull} of $\{(0,0)\} 
\cup \big\{ \big(\onu^{\phi_{\tau}}, 
\bar{\mathscr{U}}(\phi_{\tau})\big): 
\tau\in\mathbb{S}\cup\{n_s+1\} \big\}$.
Because there are a finite number of rate pairs 
of threshold policies, the lower bound of this convex hull 
is non-decreasing, piece-wise affine and convex for $\onu 
\in [0,\onu^*]$.

\subsection{Derivation of Stationary PMF in 
\eqref{eq:distributionSplit}}
	\label{appn:disSplit}

In order to prove that \eqref{eq:distributionSplit} 
is the correct stationary PMF, it suffices to show 
that the specified PMF satisfies the 
following global balance equations:
\beqa
\label{eq:distributionSplit2}
\opi^{\phi'}(\oby)
\myeq \sum_{\oby' \in \mathbb{Y}} 
	\opi^{\phi'}(\oby') \ 
	\overline{\bf P}^{\phi'}_{\oby', \oby} \ 
	\mbox{ for all } \oby \in \mathbb{Y}, 
\eeqa
where $\overline{\bf P}^{\phi'}$
is the one-step transition matrix of 
$\obY^{\phi'}$.
To this end, we shall demonstrate that the 
RHS of \eqref{eq:distributionSplit} is equal 
to the RHS of \eqref{eq:distributionSplit2}. 

First, we break the RHS of 
\eqref{eq:distributionSplit2} into two terms. 
\beqa
\hspace{-0.3in}	\sum_{\oby' \in \mathbb{Y}} 
	\opi^{\phi'}(\oby') \ 
	\overline{\bf P}^{\phi'}_{\oby', \oby} 
\myeq \overline{\pi}^{\phi'}(\tau_2-1, \cA)
	\overline{\bf P}^{\phi'}_{(\tau_2-1,\cA), \oby}
	\label{eq:split-1} \\
&& \myhb + \sum_{\oby' \in \mathbb{Y}
	\setminus \{(\tau_2-1,\cA)\} } \
	\opi^{\phi'}(\oby') \ 
	\overline{\bf P}^{\phi'}_{\oby', \oby}
	\label{eq:split-2}
\eeqa
We then rewrite each term on the RHS: 
from \eqref{eq:distributionSplit}
and \eqref{eq:form1}, we have 
\beqan
\eqref{eq:split-1}
\myeq \big(\alpha \cdot	
	\opi^{\phi_{\tau_1}}(\tau_2-1, \cA) 
	+ (1-\alpha) \opi^{\phi_{\tau_1}}(\tau_2-1,\cA)
		\big) \lb
&& \times \big( (1-\gamma) 
	\overline{\bf P}^{\phi_{\tau_1}}_{(\tau_2-1,
		\cA), \oby}
	+ \gamma 
	\overline{\bf P}^{\phi_{\tau_2}}_{(\tau_2-1,
		\cA), \oby} \big) .
\eeqan
Substituting the expression in \eqref{eq:gamma}
for $\gamma$,
\beqa
\eqref{eq:split-1}
\myeq (1-\alpha )\opi^{\phi_{\tau_1}}(\tau_2-1,\cA)
	\overline{\bf P}^{\phi_{\tau_1}}_{(\tau_2-1,
		\cA), \oby} 
	\label{eq:split-1a} \\
&& + \alpha \cdot \opi^{\phi_{\tau_2}}(\tau_2-1, 
	\cA) 
	\overline{\bf P}^{\phi_{\tau_2}}_{(\tau_2-1,
	\cA), \oby} .
	\nonumber
\eeqa
Second, from \eqref{eq:distributionSplit}
\beqan
\myb \eqref{eq:split-2}
\myeq \sum_{\oby' \in \mathbb{Y}
		\setminus \{(\tau_2-1,\cA)\} } 
	\big( \alpha \cdot \opi^{\phi_{\tau_2}}(\oby') \lb
&& \hspace{0.8in} 	
    + (1-\alpha)\opi^{\phi_{\tau_1}}(\oby')
			\big)
	\overline{\bf P}^{\phi'}_{\oby', \oby}. 
\eeqan
From \eqref{eq:form1}, for all $\oby'
= (\os', \ow') \in 
\mathbb{Y} \setminus \{ (\tau_2 - 1, 
\cA) \}$, we have 
$\phi'(\by')
= \phi_{\tau_1}(\by')$ and
$\overline{\bf P}^{\phi'}_{\oby', \oby}
= \overline{\bf P}^{\phi_{\tau_1}}_{\oby', 
\oby}$. Moreover, because $\phi_{\tau_2}$
is a deterministic policy with a threshold
on the activity state of the
server, $\opi^{\phi_{\tau_2}}(\oby')
= 0$ for all $\oby' = (\os', \ow')$ with 
$\os' < \tau_2 - 1$. 
Hence, for all $\oby' \in \mathbb{Y}
\setminus \{(\tau_2 - 1), \cA) \}$ with 
$\opi^{\phi_{\tau_2}}(\oby') > 0$, 
together with the assumption $\tau_1 \leq \tau_2$, 
we have 
\beqan
\phi_{\tau_1}(\oby') 
= \phi_{\tau_2}(\oby')
= \begin{cases}
	0 & \mbox{if } \os' \geq \tau_2 
		\mbox{ and } \ow' = \cA \\
	1 & \mbox{if } \ow' = \cB \\
	\end{cases}
\eeqan
and, consequently, 
$\overline{\bf P}^{\phi_{\tau_1}}_{\oby', \oby}
= \overline{\bf P}^{\phi_{\tau_2}}_{\oby', \oby}$. 
Therefore, 
\beqa
\eqref{eq:split-2}
\myeq \sum_{\oby' \in \mathbb{Y}
		\setminus \{(\tau_2-1,\cA)\} } 
	\Big( \alpha \cdot \opi^{\phi_{\tau_2}}(\oby')
		\overline{\bf P}^{\phi_{\tau_2}}_{\oby',
			\oby} 
	\label{eq:split-2a} \\
&& \hspace{0.9in} 
	+ (1-\alpha) \opi^{\phi_{\tau_1}}(\oby')
	\overline{\bf P}^{\phi_{\tau_1}}_{\oby', 
		\oby} \Big) .
	\nonumber
\eeqa

Substituting the new expressions in 
\eqref{eq:split-1a} and \eqref{eq:split-2a} for
\eqref{eq:split-1} and \eqref{eq:split-2}, 
respectively, we obtain
\beqan
\sum_{\oby' \in \mathbb{Y}} 
	\opi^{\phi'}(\oby') \ 
	\overline{\bf P}^{\phi'}_{\oby', \oby} 
\myeq \alpha \cdot \opi^{\phi_{\tau_2}}(\oby)
	+ (1-\alpha)\opi^{\phi_{\tau_1}}(\oby), 
\eeqan
where the equality follows from the 
fact that $\opi^{\phi_{\tau_1}}$ and 
$\opi^{\phi_{\tau_2}}$ are the stationary
PMFs of $\obY^{\phi_{\tau_1}}$
and $\obY^{\phi_{\tau_2}}$, respectively. 

\subsection{Proof of Lemma~\ref{lem:potential}}
\label{subsec:ProofOfLemmaPotential}

We shall first construct a temporary function 
$f$ that will be used to 
construct a potential function 
satisfying all conditions in 
the lemma.
For the simplicity of exposition, 
suppose that the states in $\mathbb{M}$
are ordered in some arbitrary fashion
and let $n^\star$ be some state belonging to 
the recurrent communicating 
class. First, assign $f({n^\star}) = 0$. 
Next, for each $m \in \mathbb{M}
\setminus \{n^\star\} =: \mathbb{M}^-$
(with $n_M^- := |\mathbb{M}^-|$), we rewrite 
the constraints in \eqref{eq:PF1}
as follows.
\beqa
&& \hspace{-0.3in} 
f(m) - \EX\left[f(M_{k+1}) 
		\ | \ M_{k} = m \right] \lb
\myeq f(m) 
- \sum_{\hat{m} \in \mathbb{M}}
	f(\hat{m}) 
	P_{M_{k+1} | M_{k}} ( \hat{m}
		\ | \ m) \lb
\myeq \big( 1 - P_{M_{k+1} | M_{k}} (m 
	\ | \ m) \big) f(m) \lb 
&& - \sum_{\hat{m} \in \mathbb{M}^- \setminus \{m\} } 
	f(\hat{m})
	P_{M_{k+1} | M_{k}} ( \hat{m} \ | \ m) \lb
\myeq r_{avg} - \EX \Big[ \mathcal{R}(M_{k+1},M_k) \ | \ M_k = m \Big] 
    =: \xi_m
	\label{eq:ftmp_constr}
\eeqa
These constraints can be put in
a matrix form as follows:
\beqa
{\bf B} \ {\bf f} = \boldsymbol{\xi}
    \label{eq:lemma9-1}
\eeqa
where ${\bf f} = \big( f(m): m \in \mathbb{M}^- \big)$ and  
$\boldsymbol{\xi} = \big( \xi_m: m \in \mathbb{M}^- \big)$
are $n_M^-$-dimensional column vectors, 
and ${\bf B}$ is an $n_M^- \times n_M^- $ matrix whose elements are given 
by
\beqan
\quad {\bf B}_{j, l} 
= \begin{cases}
    1 - P_{M_{k+1}|M_k}(j \ | \ j) & \mbox{if } j = l \\
    - P_{M_{k+1}|M_k}(l \ | \ j) & \mbox{if } j \neq l
\end{cases}, \quad j, l \in \mathbb{M}^-.
\eeqan

To complete the proof, we need
the following lemma. 

\begin{lemma}
The matrix ${\bf B}$ is 
weakly chained diagonally dominant.
\end{lemma}
\begin{proof} 
First, the matrix is 
weakly diagonally dominant because
\beqan
\big| {\bf B}_{j, j} \big|
\myeq 1 - P_{M_{k+1} | M_k}( j \ | \ j ) = \sum_{l \in \mathbb{M} \setminus \{j\} }
	P_{M_{k+1} | M_k} ( l \ | \ j) \lb
\mygeq \sum_{l \in \mathbb{M}^- \setminus \{ j \} }
	P_{M_{k+1} | M_k} ( l \ | \  j) = \sum_{l \in \mathbb{M}^- \setminus 
	\{j \} } \big| {\bf B}_{j, l} \big|. 
\eeqan
Second, for any state $j$ in $\mathbb{M}^-$ with 
$P_{M_1 | M_0} ( {n^\star} \ | \ j) > 0$, 
the $j$th row of ${\bf B}$ is strictly diagonally 
dominant (SDD) because
\beqan
\big| {\bf B}_{j,j} \big|
> \sum_{l \in \mathbb{M}^- \setminus \{j\} }
	\big| {\bf B}_{j, l} \big|. 
\eeqan

Finally, note that the $j$th row of ${\bf B}$, $j 
\in \mathbb{M}^-$, is 
not SDD only if $P_{M_1 | M_0} ( n^\star \ | \ j) = 0$. 
Suppose that there exists
a row of ${\bf B}$, say the $l'$th row, which 
is not SDD. Then, since $M$
has only one recurrent communicating class, there exist (i) some
$l^+$ in $\mathbb{M}^-$ 
such that the $l^+$th row is SDD and (ii) 
a path from state $l'$ to state 
${l^+}$ in the directed graph associated
with matrix ${\bf B}$. This proves that the
matrix ${\bf B}$ is weakly chained diagonally
dominant. 
\end{proof}

Since weakly chained diagonally dominant matrix 
is nonsingular~\cite{azimzadeh2016weakly}, there
is a unique solution to the set of 
linear equations in \eqref{eq:lemma9-1}, which 
is then assigned to the temporary function $f(m)$, 
$m \in \mathbb{M}^-$.
Recall that, by construction, the function 
$f$ satisfies the condition \eqref{eq:PF1} at
all states $m$ in $\mathbb{M}^-$.
We now prove that the condition 
\eqref{eq:PF1} holds at state $n^*$ as well.

First, we can show that the following equality
holds:
\beqa
&& \hspace{-0.5in}
r_{avg} = 
\sum_{m \in \mathbb{M}} \varrho^M(m)
	\Big(\EX \big[ \mathcal{R}(M_{k+1},M_k) 
		\ | \ M_k = m \big]
	\label{eq:ravg2} \\
&& \hspace{0.3in} - \EX \big[ f(M_{k+1}) - f(M_{k}) 
	\ | \ M_k = m \big] \Big), 
	\nonumber
\eeqa
where $\varrho^M$ is the stationary PMF of 
$M$. From the definition of stationary PMF, 
\beqan
&& \hspace{-0.3in}
\sum_{m \in \mathbb{M}} \varrho^M(m)
\EX \big[ f(M_{k+1}) \ | \ M_k = m \big]  \lb  
\myeq \sum_{m \in \mathbb{M}} \varrho(m)
    \Big( \sum_{m' \in \mathbb{M}} f(m') 
        P_{M_{k+1} | M_k}(m' \ | \ m) \Big) \lb 
\myeq \sum_{m' \in \mathbb{M}} f(m')
    \Big( \sum_{m \in \mathbb{M}} \varrho(m)
        P_{M_{k+1} | M_k}(m' \ | \ m) \Big) \lb 
\myeq \sum_{m' \in \mathbb{M}} f(m') \varrho(m').
\eeqan
Therefore, we get
\beqan
\sum_{m \in \mathbb{M}} \varrho^M(m)
	\Big( \EX \big[ f(M_{k+1}) - f(M_{k}) 
	\ | \ M_k = m \big] \Big) = 0, 
\eeqan
and the equality in \eqref{eq:ravg2} follows
from the definition of $r_{avg}$.

Rewriting \eqref{eq:ravg2} using the equality in 
\eqref{eq:ftmp_constr}, we obtain 
\beqan
r_{avg} 
\myeq \varrho^M(n^\star)
	\Big(\EX \big[ \mathcal{R}(M_{k+1},M_k) \ | \ 
	    M_k = n^\star \big]\lb
	&& \quad - \EX \big[ f(M_{k+1}) - f(M_{k}) \ | \ 
	    M_k = n^\star \big] \Big)\lb
	&& \myhb + \sum_{m \in \mathbb{M}^- }
		\varrho^M(m) \ r_{avg}.
\eeqan
Moving the second
term on the RHS to the LHS, we obtain
\beqan
&& \hspace{-0.3in}
\Big(1 - \sum_{m \in \mathbb{M}^- }
		\varrho^M(m) \Big) r_{avg}
= \varrho^M(n^\star) \ r_{avg} \lb
\myeq \varrho^M(n^\star)
	\Big(\EX \big[ \mathcal{R}(M_{k+1},M_k) \ | \ M_k = n^\star \big]\lb
	&& - \EX \big[ f(M_{k+1}) - f(M_{k}) \ | \ M_k = n^\star \big] \Big).
\eeqan
Thus, we have the desired equality 
$\EX \big[ \mathcal{R}(M_{k+1},M_k) \ | \ M_k = n^\star \big]$
$- \EX \big[ f(M_{k+1}) - f(M_{k}) \ | \ M_k = n^\star \big]$ 
$= r_{avg}$ 
because $\varrho^M(n^\star)>0$ from the 
assumption that $n^\star$ belongs to the unique positive 
recurrent communicating class. 

Finally, we define the nonnegative 
potential-like function 
$\mathscr{H}: \mathbb{M} \to \R_+$, where 
$\mathscr{H}(m) = f(m) - \min_{m' \in 
\mathbb{M}} f(m')$
for all $m \in \mathbb{M}$. From its
construction, the
function $\mathscr{H}$ is non-negative and 
satisfies all the constraints in the 
lemma.

\subsection{Auxiliary Results and Proof of Theorem~\ref{thm:MainDistConv}}
\label{app:ProofOfMainDistConv}

We make use of the following lemmas 
(Lemmas \ref{lem:DistConverge1} through
\ref{lem:DistConverge3}) to complete the
proof of the theorem. Let $\overline{\bf P}^{\phi}$
and $\obpi^\phi$ be the one-step 
transition matrix and the stationary PMF
(given as a row vector), respectively, of 
$\overline{\bY}^{\phi}$.
Recall that we defined 
$\Phi^\epsilon_{R}$ to be the set of $\phi\in\Phi_R$ 
such that $\phi(1,\cA) \geq \epsilon$ and that,  
for any $\onu$ in $(\lambda,\onu^*)$, the set 
$\Phi_R^{\epsilon}(\onu)$ is nonempty by 
Proposition~\ref{prop:PhiEpsilonNonEmpty}.

\begin{lemma}
\label{lem:DistConverge1}
There exists a positive constant $\eta_\epsilon$ such that, for any distribution ${\bf p}$ over $\mathbb{Y}$, we have
\beqan
\sum_{r=1}^\infty \left\Vert{\bf p} 
	\big(\overline{\bf P}^{\phi} \big)^r 
	- \obpi^\phi \right\Vert_1 \myleq \eta_\epsilon 
	\ \mbox{ for all } \phi \in \Phi_R^{\epsilon}.
\eeqan
\end{lemma}
\begin{proof}
First, we can find positive $\tilde{\alpha}_\epsilon$
such that, for all $\phi' \in \Phi^\epsilon_{R}$, 
\begin{equation}
\label{eq:lowerbdCondPAuxilary}
P \Big ( \obY_{k+2n_s}^{\phi'} = (n_s,\cB) 
	\ \big | \ \obY_{k}^{\phi'} = \by \Big ) 
\geq \tilde{\alpha}_{\epsilon}, 
	\quad \by \in \mathbb{Y}. 
\end{equation}
One can verify that, for example, 
\begin{eqnarray*}
\label{eq:AlphaTildeDef}
\myb \tilde{\alpha}_{\epsilon} 
= \epsilon (1-\mu(n_s))^{2n_s}  
	\Pi_{s=1}^{n_s-1} (1-\mu(s)) \rho_{s+1,s}\rho_{s,s+1}
\end{eqnarray*}
satisfies the inequality in \eqref{eq:lowerbdCondPAuxilary}.

Next, we follow an analysis that is similar to the proof of Theorem 4.16 of \cite{seneta2006non}. We define a function $h:\mathbb{R}^{2 n_S \times 
2 n_S} \rightarrow \mathbb{R}_+$ with
\beqan
h({\bf P}) 
= \frac{1}{2}\max_{i,j}\sum_{\ell=1}^{2 n_S} 
    \big| {\bf P}_{i, \ell} - {\bf P}_{j, \ell} \big|, 
\eeqan
where ${\bf P}_{i, \ell}$ is the element in the $i$th
row and $\ell$th column of matrix ${\bf P}$.

Note that since $\phi \in \Phi^{\epsilon}_R$, by
\eqref{eq:lowerbdCondPAuxilary} every element in the column of $\big(\overline{\bf P}^{\phi}\big)^{2n_s}$ corresponding to $(n_s,\cB)$ is lower-bounded by some positive $\tilde{\alpha}_{\epsilon}$. Thus, equation (4.6) of \cite{seneta2006non} tells us 
\beqan
 h\Big(\big(\overline{\bf P}^{\phi}\big)^{2n_s}\Big) \leq 1-\tilde{\alpha}_{\epsilon}.
\eeqan
Proceeding with the proof, for every $r\geq2n_s$ and  $\kappa = \floor*{r/2n_s}$,
\beqan
&& \hspace{-0.3in} 
h\Big(\Big(\overline{\bf P}^{\phi}\Big)^{r}\Big) 
=  h\Big(\big(\overline{\bf P}^{\phi}\big)^{r-2 \kappa n_s}
    \big(\overline{\bf P}^{\phi}\big)^{2 \kappa n_s}\Big) \lb
\myleq 
     h\Big(\big(\overline{\bf P}^{\phi}\big)^{r-2 \kappa n_s}\Big)
        \Big( h\Big(\big(\overline{\bf P}^{\phi}\big)^{2n_s}\Big)
            \Big)^\kappa \lb
\myleq  
    \Big(1-\tilde{\alpha}_{\epsilon}\Big)^{\kappa} 
\leq \Big(1-\tilde{\alpha}_{\epsilon}\Big)^{\frac{r}{2n_s}-1}= K_\epsilon\sigma_\epsilon^r,
\eeqan
where $K_\epsilon = (1-\tilde{\alpha}_{\epsilon})^{-1}$ and $\sigma_\epsilon = (1-\tilde{\alpha}_{\epsilon})^{1/2n_s}$. The first inequality follows from Lemma 4.3 of \cite{seneta2006non}, 
which states
$h({\bf P} {\bf P}^\dagger) \leq  h({\bf P}) h({\bf P}^\dagger)$ for any two stochastic
matrices ${\bf P}$ and ${\bf P}^\dagger$. The second inequality follows from the observation $ h({\bf P}) \leq 1$ for any stochastic matrix ${\bf P}$, and $\Big(1-\tilde{\alpha}_{\epsilon}\Big) < 1$ leads to the final inequality.
Combining with Lemma 4.3 of \cite{seneta2006non} and the fact that the sum of all elements of ${\bf p}-\obpi^{\phi}$ equals zero, we know that, for every $r\geq2n_s$,
\beqa \label{eq:lem8-1}
\left\Vert{\bf p} \ 
	\big(\overline{\bf P}^{\phi} \big)^r 
	- \obpi^\phi \right\Vert_1
=
    \left\Vert{\bf p} \ 
	\big(\overline{\bf P}^{\phi} \big)^r 
	- \obpi^\phi \big(\overline{\bf P}^{\phi} \big)^r  \right\Vert_1      \lb
\leq
     h\Big(\big(\overline{\bf P}^{\phi}\big)^{r}\Big)
    \left\Vert{\bf p}
	- \obpi^\phi\right\Vert_1
\leq
    2K_\epsilon\sigma_\epsilon^r.
\eeqa
Hence, the inequality in \eqref{eq:lem8-1}
yields the following bound:
\beqan
&&\myhb \sum_{r=1}^\infty \left\Vert{\bf p} \ 
	\big(\overline{\bf P}^{\phi} \big)^r 
	- \obpi^\phi \right\Vert_1\lb
\myeq \sum_{r=1}^{2n_s} \left\Vert{\bf p} \ 
	\big(\overline{\bf P}^{\phi} \big)^r 
	- \obpi^\phi \right\Vert_1 
	+ \sum_{r=2n_s+1}^{\infty} \left\Vert{\bf p} \ 
	\big(\overline{\bf P}^{\phi} \big)^r 
	- \obpi^\phi \right\Vert_1\lb
\myleq 4n_s
	+ \sum_{r=2n_s+1}^{\infty} 2K_\epsilon\sigma_\epsilon^r
= 4n_s
	+ \frac{2K_\epsilon\sigma_\epsilon^{2n_s+1}}{1-\sigma_\epsilon}=:\eta_\epsilon
\eeqan
\end{proof}

Define $\bvarrho^{\mathscr{X}(\phi)}$ to be a row vector representing
the distribution of server state under the 
stationary PMF $\pi^{\mathscr{X}(\phi)}$ of  $\bX^{\mathscr{X}(\phi)}$, which is given by 
\beqa
\label{eq:varrhoDef}
\varrho^{\mathscr{X}(\phi)}(\by) \myeq \sum_{q \in \mathbb{Q}^{\overline{\by}}} \pi^{\mathscr{X}(\phi)}(\by,q), \quad \by\in\mathbb{Y}.
\eeqa

\begin{lemma}
\label{lem:DistConverge2}
Fix $\onu$ in $(\lambda, \onu^*)$ and $\epsilon$
in $(0,1)$, and let $\beta_{\lambda, \epsilon}$ 
be a positive constant satisfying 
Lemma~\ref{lem:QGoesToZero}. 
Then, the following bound holds for every 
$\phi$ in $\Phi_R^{\epsilon}(\onu)$ and 
all $r$ in $\N$:
\beqan
\left\Vert \bvarrho^{\mathscr{X}(\phi)} - \bvarrho^{\mathscr{X}(\phi)}\big(\overline{\bf P}^{\phi} \big)^r \right\Vert_1 \myleq 2r\frac{(\onu-\lambda)}{\beta_{\lambda,\epsilon}}
\eeqan
\end{lemma}
\begin{proof}
Let $\overline{\bf P}^{\phi_1}$ 
be the one-step
transition matrix of $\overline{\bY}$ 
under the policy $\phi_1$, which always chooses 
$\mathcal{R}$ when the server is available.
We denote the row of $\overline{\bf P}^{\phi}$ 
(resp. $\overline{\bf P}^{\phi_1}$) 
corresponding to the server state $\by 
= (s, w) \in 
\mathbb{Y}$ by $\overline{\bf P}^{\phi}_{\by}$
(resp. $\overline{\bf P}^{\phi_1}_{\by}$).

Since $\bvarrho^{\mathscr{X}(\phi)}$ 
remains the same after one step transition, using 
the equality in \eqref{eq:varrhoDef}, we can rewrite
$\bvarrho^{\mathscr{X}(\phi)}$ as 
\beqa
&& \hspace{-0.55in} 
\bvarrho^{\mathscr{X}(\phi)}
= \sum_{s \in \mathbb{S}} \Big[
	\pi^{\mathscr{X}(\phi)}(s, \cA, 0)
			\ \overline{\bf P}^{\phi_1}_{(s,\cA)} \lb
&& \hspace{0.35in}
	+ \sum_{w \in \mathbb{W}} \Big( \sum_{q = 1}^\infty 
	\pi^{\mathscr{X}(\phi)}(s, w, q)  \Big)
			\overline{\bf P}^\phi_{(s,w)} \Big] \lb 
&& \hspace{-0.55in} 
= \sum_{s \in \mathbb{S}}
	\Big[ \pi^{\mathscr{X}(\phi)}(s, \cA, 0) 
			\ \overline{\bf P}^{\phi_1}_{(s, \cA)} 
	+\varrho^{\mathscr{X}(\phi)}(s, \cB) 
		\ \overline{\bf P}^\phi_{(s, \cB)} \lb
&& + \Big( \varrho^{\mathscr{X}(\phi)}(s, \cA)
	- \pi^{\mathscr{X}(\phi)}(s, \cA, 0) \Big)
			\overline{\bf P}^\phi_{(s, \cA)} \Big] \lb
&& \hspace{-0.55in} 
= \sum_{s \in \mathbb{S}} 
	\Big[ \pi^{\mathscr{X}(\phi)}(s, \cA, 0) 
		\big( \overline{\bf P}^{\phi_1}_{(s, \cA)} 
			- \overline{\bf P}^{\phi}_{(s, \cA)}  \big)
					\Big] 
+ \bvarrho^{\mathscr{X}(\phi)} \ 
	\overline{\bf P}^\phi.
		\label{eq:DC3-1}
\eeqa

Define 
$\boldsymbol{\gamma}^{\phi}
\Eqdef \sum_{s \in \mathbb{S}} 
	\big[ \pi^{\mathscr{X}(\phi)}(s, \cA, 0) 
		\big( \overline{\bf P}^{\phi_1}_{(s, \cA)} 
			- \overline{\bf P}^{\phi}_{(s, \cA)}  \big)
					\big]$. 
Applying (\ref{eq:DC3-1}) iteratively, we
obtain 
\beqa
\bvarrho^{\mathscr{X}(\phi)}
\myeq \bvarrho^{\mathscr{X}(\phi)}
	\big( \overline{\bf P}^\phi \big)^r 
	+ \boldsymbol{\gamma}^{\phi}\sum_{\tau=1}^{r} 
		\big( \overline{\bf P}^\phi \big)^{\tau-1}, 
	\ r \in \N.
	\label{eq:DC3-2}
\eeqa
Subtracting the first term on the RHS
of (\ref{eq:DC3-2}) from both sides and taking
the norm, 
\beqan
&& \hspace{-0.3in}
\left\Vert \bvarrho^{\mathscr{X}(\phi)}
	- \bvarrho^{\mathscr{X}(\phi)}
		\big( \overline{\bf P}^\phi \big)^r  
	 		\right\Vert_1 
= \left\Vert \boldsymbol{\gamma}^{\phi}\sum_{\tau=1}^{r} 
		\big( \overline{\bf P}^\phi \big)^{\tau-1}
			\right\Vert_1 \lb
\myleq \left\Vert
    \boldsymbol{\gamma}^{\phi}
    \right\Vert_1
        \sum_{\tau=1}^{r} 
	\left\Vert 
		\big( \overline{\bf P}^\phi \big)^{\tau-1}
			\right\Vert_\infty 
= r\left\Vert
    \boldsymbol{\gamma}^{\phi}
    \right\Vert_1, 
	\label{eq:DC3-3}
\eeqan
where the last equality is a consequence
of $\left\Vert {\bf P} \right\Vert_\infty = 1$
for a stochastic matrix ${\bf P}$.
Substituting the expression for $\gamma^{\phi}$
and using the inequality 
$\norm{\overline{\bf P}^{\phi_1}_{\by}
- \overline{\bf P}^\phi_{\by}}_1 \leq 2$ for all 
$\by \in \mathbb{Y}$, we get
\beqan
r\left\Vert
    \boldsymbol{\gamma}^{\phi}
    \right\Vert_1
\myleq 2 r \Big( \sum_{s \in \mathbb{S}}
	\pi^{\mathscr{X}(\phi)}(s, \cA, 0) \Big).
	\label{eq:DC3-4}
\eeqan
Thus, we get
\beqan
\left\Vert \bvarrho^{\mathscr{X}(\phi)} - \bvarrho^{\mathscr{X}(\phi)}\big(\overline{\bf P}^{\phi} \big)^r \right\Vert_1
\myleq 2 r \Big( \sum_{s \in \mathbb{S}}
	\pi^{\mathscr{X}(\phi)}(s, \cA, 0) \Big)\lb
\myleq 2r\frac{(\onu-\lambda)}{\beta_{\lambda,\epsilon}}.
\eeqan
The last inequality holds because $\beta_{\lambda, \epsilon}$ 
satisfies Lemma \ref{lem:QGoesToZero}.
\end{proof}

\begin{lemma}
\label{lem:DistConverge3}
Fix $\onu$ in $(\lambda, \onu^*)$ and $\epsilon$
in $(0,1)$, and let $\beta_{\lambda, \epsilon}$ 
and $\eta_\epsilon$ be positive constants satisfying 
Lemmas~\ref{lem:QGoesToZero} and
\ref{lem:DistConverge2}, respectively. 
Then, the following inequality holds for every $N\in\N$:
\beqan
\left\Vert \bvarrho^{\mathscr{X}(\phi)} - \obpi^\phi \right\Vert_1 \myleq \frac{\eta_{\epsilon}}{N}+\frac{(N+1)(\onu - \lambda)} {\beta_{\lambda,\epsilon}}, \ \phi \in \Phi_R^{\epsilon}(\onu)
\eeqan
\end{lemma}
\begin{proof}
The proof of the lemma is a simple application
of Lemmas~\ref{lem:DistConverge1} and
\ref{lem:DistConverge2}. 
\beqan
\left\Vert \bvarrho^{\mathscr{X}(\phi)} 
	- \obpi^\phi \right\Vert_1 
\myleq \frac{1}{N}\sum_{r = 1}^N 
	\left\Vert \bvarrho^{\mathscr{X}(\phi)} 
	- \bvarrho^{\mathscr{X}(\phi)}
		\big(\overline{\bf P}^{\phi} \big)^r
		\right\Vert_1\lb
    && \myb +\frac{1}{N} \sum_{r = 1}^N 
    	\left\Vert \bvarrho^{\mathscr{X}(\phi)}
    	\big(\overline{\bf P}^{\phi} \big)^r
    	- \obpi^\phi \right\Vert_1 \lb
\myleq \frac{(N+1)(\onu - \lambda)}
	{\beta_{\lambda,\epsilon}} 
	+ \frac{\eta_{\epsilon}}{N},
\eeqan
where the last inequality utilizes $\sum_{r=1}^N r
= N (N+1) / 2$.
\end{proof}

\subsubsection{\it Proof of Theorem~\ref{thm:MainDistConv}}

We have
\beqan
&& \hspace{-0.3in}
\sum_{\by\in\mathbb{Y}} \left|\opi^\phi(\by) - \sum_{q>0} \pi^{\mathscr{X}(\phi)}(\by,q) \right|\lb
\myleq
    \left\Vert \obpi^\phi - \bvarrho^{\mathscr{X}(\phi)} \right\Vert_1 +\sum_{\by\in\mathbb{Y}}  \left|\varrho^{\mathscr{X}(\phi)}(\by)- \sum_{q>0} \pi^{\mathscr{X}(\phi)}(\by,q) \right|\lb
\myeq
    \left\Vert \obpi^\phi - \bvarrho^{\mathscr{X}(\phi)} \right\Vert_1 + \sum_{s\in\mathbb{S}}  \pi^{\mathscr{X}(\phi)}(s,\cA,0)\lb
\myleq
    \frac{\eta_{\epsilon}}{N}+\frac{(N+1)(\onu - \lambda)}{\beta_{\lambda,\epsilon}}+\frac{\onu-\lambda}{\beta_{\lambda,\epsilon}},
\eeqan
where the final inequality follows from Lemmas \ref{lem:QGoesToZero} and 
\ref{lem:DistConverge3}. By selecting $N = \ceil*{\frac{\eta_\epsilon}{(\onu-\lambda)^{\frac{1}{2}}}}$, we obtain
the inequality in \eqref{eq:MainDistConvIneq}:
\beqan
&& \hspace{-0.3in}
\sum_{\by\in\mathbb{Y}} \left|\opi^\phi(\by) - \sum_{q>0} \pi^{\mathscr{X}(\phi)}(\by,q) \right|\lb
\myleq
    (\onu-\lambda)^{\frac{1}{2}}+\frac{\Big(\frac{\eta_\epsilon}{(\onu-\lambda)^{\frac{1}{2}}}+1+1\Big)(\onu - \lambda)}{\beta_{\lambda,\epsilon}}+\frac{\onu-\lambda}{\beta_{\lambda,\epsilon}}\lb
\myleq
    \frac{\beta_{\lambda,\epsilon} + \eta_\epsilon}{\beta_{\lambda,\epsilon}}(\onu-\lambda)^{\frac{1}{2}}+\frac{3}{\beta_{\lambda,\epsilon}}(\onu-\lambda)
\eeqan

\bibliographystyle{ieeetr}
\bibliography{martinsrefs}

\end{document}